\documentclass[10pt]{article}
\usepackage{amsxtra}
\usepackage{amssymb, amsmath}
\usepackage{amsthm}
\usepackage{hyperref}
\usepackage{multicol}
\usepackage{latexsym}

\newtheorem{lemma}{Lemma}[section]
\newtheorem{theorem}{Theorem}[section]
\newtheorem{remark}{Remark}[section]

\usepackage[dvips]{graphicx}

\usepackage{palatino} % com serifa

\begin{document}
\begin{center}
\textbf{\LARGE{Nested Inequalities Among Divergence Measures} }
\end{center}

\bigskip
\begin{center}
\textbf{\large{Inder J. Taneja}}\\
Departamento de Matem\'{a}tica\\
Universidade Federal de Santa Catarina\\
88.040-900 Florian\'{o}polis, SC, Brazil.\\
\textit{e-mail: ijtaneja@gmail.com\\
http://www.mtm.ufsc.br/$\sim $taneja}
\end{center}

\begin{abstract}
In this paper we have considered an inequality having 11 divergence measures. Out of them three are logarithmic such as Jeffryes-Kullback-Leiber \cite{jef} \cite{kul} \textit{J-divergence}. Burbea-Rao \cite{bur} \textit{Jensen-Shannon divergence }and Taneja \cite{tan1} \textit{arithmetic-geometric mean divergence}. The other three are non-logarithmic such as \textit{Hellinger discrimination}, \textit{symmetric }$\chi ^2 - $\textit{divergence}, and \textit{triangular discrimination}. Three more are considered are due to mean divergences. Pranesh and Johnson \cite{kuj} and Jain and Srivastava \cite{jas} studied different kind of divergence measures. We have considered measures arising due to differences of single inequality having 11 divergence measures in terms of a sequence. Based on these differences we have obtained many inequalities. These inequalities are kept as nested or sequential forms. Some reverse inequalities and equivalent versions are also studied.
\end{abstract}

\small
\begin{multicols}{2}

\section{Introduction}

Let
\[
\Gamma _n = \left\{ {P = (p_1 ,p_2 ,...,p_n )\left| {p_i > 0,\sum\limits_{i
= 1}^n {p_i = 1} } \right.} \right\},n \ge 2,
\]

\noindent
be the set of all complete finite discrete probability distributions. Let us
consider the two groups of divergence measures:

\bigskip
\noindent
\textbf{$\bullet$ Logarithmic divergence measures}
\begin{align}
I(P\vert \vert Q)&  = \frac{1}{2}\left[ {\sum\limits_{i = 1}^n {p_i \ln \left(
{\frac{2p_i }{p_i + q_i }} \right) + } \sum\limits_{i = 1}^n {q_i \ln \left(
{\frac{2q_i }{p_i + q_i }} \right)} } \right],\notag\\
J(P\vert \vert Q) & = \sum\limits_{i = 1}^n {(p_i - q_i )\ln (\frac{p_i }{q_i
})}\notag
\intertext{and}
T(P\vert \vert Q) & = \sum\limits_{i = 1}^n {\left( {\frac{p_i + q_i }{2}}
\right)\ln \left( {\frac{p_i + q_i }{2\sqrt {p_i q_i } }} \right)}.\notag
\end{align}

\noindent
\textbf{$\bullet$ Non-logarithmic divergence measures}
\begin{align}
\Delta (P\vert \vert Q) & = \sum\limits_{i = 1}^n {\frac{(p_i - q_i )^2}{p_i +
q_i }},\notag\\
h(P\vert \vert Q) & = \frac{1}{2}\sum\limits_{i = 1}^n {(\sqrt {p_i } - \sqrt
{q_i } )^2}\notag
\intertext{and}
\Psi (P\vert \vert Q) & = \sum\limits_{i = 1}^n {\frac{(p_i - q_i )^2(p_i +
q_i )}{p_i q_i }}\notag
\end{align}

The logarithmic measures $I(P\vert \vert Q)$, $J(P\vert \vert Q)$ and $T(P\vert \vert Q)$ are three classical divergence measures known in the literature on information theory and statistics are \textit{Jensen-Shannon divergence}, \textit{J-divergence} and \textit{Arithmetic-Geometric mean divergence }respectively \cite{tan1} \cite{tan3}. The non-logarithmic measures $\Delta (P\vert \vert Q)$, $h(P\vert \vert Q)$ and $\Psi (P\vert \vert Q)$ are respectively known as \textit{triangular discrimination}, \textit{Hellingar's divergence} \textit{and symmetric chi-square divergence}. In 2005, the author \cite{tan4} proved the following inequality among these six symmetric divergence measures:

\begin{equation}
\label{eq1}
\frac{1}{4}\Delta \le I \le h \le \frac{1}{8}J \le T \le \frac{1}{16}\Psi.
\end{equation}

The above inequality (\ref{eq1}) admits many nonnegative differences among the
divergence measures. Based on these non-negative differences, the author \cite{tan4}
proved the following result:
\begin{align}
D_{I\Delta } & \le \frac{2}{3}D_{h\Delta } \le \left\{ {\begin{array}{l}
 2D_{hI} \\
 \textstyle{1 \over 2}D_{J\Delta } \le \textstyle{1 \over 3}D_{T\Delta } \\
 \end{array}} \right\} \le D_{TJ} \le\notag\\
& \le \textstyle{2 \over 3}D_{Th} \le
2D_{Jh} \le  \textstyle{1 \over 6}D_{\Psi \Delta }\le\notag\\
\label{eq2}
&  \le \textstyle{1 \over 5}D_{\Psi
I} \le \textstyle{2 \over 9}D_{\Psi h} \le \textstyle{1 \over 4}D_{\Psi J}
\le \textstyle{1 \over 3}D_{\Psi T},
\end{align}

\noindent
where, for example $D_{I\Delta } : = I - \textstyle{1 \over
4}\Delta $, $D_{TJ} : = T - \textstyle{1 \over 8}J$,
$D_{F\Psi } : = \textstyle{1 \over {16}}F - \textstyle{1 \over 8}\Psi $, etc.
The proof of the inequalities (\ref{eq2}) is based on the following two lemmas:

\begin{lemma} If the function $f:[0,\infty ) \to {\rm R}$ is convex
and normalized, i.e., $f(1) = 0$, then the \textit{f-divergence}, $C_f (P\vert \vert Q)$ given by
\begin{equation}
\label{eq3}
C_f (P\vert \vert Q) =
\sum\limits_{i = 1}^n {q_i f\left( {\frac{p_i }{q_i }} \right)} ,
\end{equation}

\noindent
is nonnegative and convex in the pair of probability distribution $(P,Q) \in
\Gamma _n \times \Gamma _n $.
\end{lemma}

\begin{lemma} Let $f_1 ,f_2 :I \subset {\rm R}_ + \to {\rm R}$ two
generating mappings are normalized, i.e., $f_1 (1) = f_2 (1) = 0$ and
satisfy the assumptions:

\noindent (i) $f_1 $ and $f_2 $ are twice differentiable on $(a,b)$;

\noindent (ii) there exists the real constants $\alpha ,\beta $such that $\alpha <
\beta$ and
\[
\alpha \le \frac{f_1 ^{\prime \prime }(x)}{f_2 ^{\prime \prime }(x)} \le
\beta, \, f_2 ^{\prime \prime }(x) > 0, \, \forall x \in (a,b),
\]

\noindent
then we have the inequalities:
\[
\alpha \mbox{ }C_{f_2 } (P\vert \vert Q) \le C_{f_1 } (P\vert \vert Q) \le
\beta \mbox{ }C_{f_2 } (P\vert \vert Q).
\]
\end{lemma}

The Lemma 1.1 is due to Csisz\'{a}r \cite{csi} and the Lemma 1.2 is due to
author \cite{tan3}. The aim of this paper is to consider more measures in
(\ref{eq1}) and improve the inequalities given in (\ref{eq2}). These measures are based
on the some well-known mean divergences.

\subsection{Mean Divergence Measures}

Author \cite{tan2} studied the following inequality
\begin{equation}
\label{eq4}
G(P\vert \vert Q) \le N_1 (P\vert \vert Q) \le N_2 (P\vert \vert Q) \le
A(P\vert \vert Q),
\end{equation}

\noindent where
\begin{align}
& G(P\vert \vert Q) = \sum\limits_{i = 1}^n {\sqrt {p_i q_i } },\notag\\
& N_1 (P\vert \vert Q) = \sum\limits_{i = 1}^n {\left(\frac{\sqrt {p_i } + \sqrt {q_i } }{2}
\right)}^2,\notag\\
& N_2 (P\vert \vert Q) = \sum\limits_{i = 1}^n {\sqrt {\frac{p_i + q_i }{2}}
\left( {\frac{\sqrt {p_i } + \sqrt {q_i } }{2}} \right)}\notag
\intertext{and}
& A(P\vert \vert Q) = \sum\limits_{i = 1}^n {\frac{p_i + q_i }{2}} = 1.\notag
\end{align}

The above inequality admits non-negative differences given by
\begin{align}
M_1 & (P\vert \vert Q) = D_{N_2 N_1 } (P\vert \vert Q)=\notag\\
& = \sum\limits_{i = 1}^n
{\left( {\sqrt {\frac{p_i + q_i }{2}} \left( {\frac{\sqrt {p_i } + \sqrt
{q_i } }{2}} \right) - \left( {\frac{\sqrt {p_i } + \sqrt {q_i } }{2}}
\right)^2} \right)},\notag\\
M_2 & (P\vert \vert Q) = D_{N_2 G} (P\vert \vert Q)=\notag\\
& = \sum\limits_{i = 1}^n
{\left[ {\left( {\frac{\sqrt {p_i } + \sqrt {q_i } }{2}} \right)\sqrt
{\frac{p_i + q_i }{2}} - \sqrt {p_i q_i } } \right]},\notag\\
M_3 & (P\vert \vert Q) = D_{AN_2 } (P\vert \vert Q)=\notag\\
& = \sum\limits_{i = 1}^n
{\left[ {\left( {\frac{p_i + q_i }{2}} \right) - \left( {\frac{\sqrt {p_i }
+ \sqrt {q_i } }{2}} \right)\left( {\sqrt {\frac{p_i + q_i }{2}} } \right)}
\right]}\notag
\intertext{and}
h & (P\vert \vert Q) = 2 D_{AN_1 } (P\vert \vert Q)=D_{AG} (P\vert \vert Q)
= D_{N_1 G} (P\vert \vert Q).\notag
\end{align}

\subsection{New Measures}

Jain and Srivastava \cite{jas} and Kumar and Johnson \cite{kuj} respectively studied the
measures
\[
K_0 (P\vert \vert Q) = \sum\limits_{i = 1}^n {\frac{(p_i - q_i )^2}{\sqrt
{p_i q_i } }}
\]
\noindent{and}
\[
F(P\vert \vert Q) = \frac{1}{2}\sum\limits_{i = 1}^n {\frac{(p_i^2 - q_i^2
)^2}{\sqrt {(p_i q_i )^3} }}.
\]

In total we have 11 divergence measures. By the application of Lemmas 1
and 2 we can put them in a single inequality as
\[
\textstyle{1 \over 4}\Delta  \le I \le 4M_1
 \le \textstyle{4 \over 3}M_2  \le h \le 4\,M_3 \le
\]
\begin{equation}
\label{eq5}
\hspace{10pt} \le \textstyle{1 \over 8}J \le T \le
\textstyle{1 \over 8}K_0  \le \textstyle{1 \over {16}}\Psi
 \le \textstyle{1 \over {16}}F.
\end{equation}

\subsection{Pyramid}

The 11 measures appearing in the inequalities (\ref{eq5}) admits 55 non-negative
differences. These 55 non-negative difference satisfies some
obvious inequalities given below in the form of \textbf{pyramid} or \textbf{triangular}:

\begin{enumerate}
\item $D_{I\Delta }^1$;\\
\item $D_{M_1 I}^2 \le D_{M_1 \Delta }^3$;\\
\item $D_{M_2 M_1 }^4 \le D_{M_2 I}^5 \le D_{M_2 \Delta }^6$;\\
\item $D_{hM_2 }^7 \le D_{hM_1 }^8 \le D_{hI}^9 \le D_{h\Delta }^{10}$;\\
\item $D_{M_3 h}^{11} \le D_{M_3 M_2 }^{12} \le D_{M_3 M_1 }^{13} \le D_{M_3
I}^{14} \le D_{M_3 \Delta }^{15}$;\\
\item $D_{JM_3 }^{16} \le D_{Jh}^{17} \le D_{JM_2 }^{18} \le D_{JM_1 }^{19} \le
D_{JI}^{20} \le D_{J\Delta }^{21}$;\\
\item $D_{TJ}^{22} \le D_{TM_3 }^{23} \le D_{Th}^{24} \le D_{TM_2 }^{25} \le
D_{TM_1 }^{26} \le D_{TI}^{27} \le \\
\le D_{T\Delta }^{28}$;\\
\item $D_{K_0 T}^{29} \le D_{K_0 J}^{30} \le D_{K_0 M_3 }^{31} \le D_{K_0 h}^{32}
\le D_{K_0 M_2 }^{33} \le  \\ \le D_{K_0 M_1 }^{34} \le D_{K_0 I}^{35}
\le D_{K_0 \Delta }^{36}$;\\
\item $D_{\Psi K_0 }^{37} \le D_{\Psi T}^{38} \le D_{\Psi J}^{39} \le D_{\Psi M_3
}^{40} \le D_{\Psi h}^{41} \le D_{\Psi M_2 }^{42} \le \\ \le D_{\Psi M_1 }^{43} \le D_{\Psi I}^{44}
\le D_{\Psi \Delta }^{45}$;\\
\item $D_{F\Psi }^{46} \le D_{FK_0 }^{47} \le D_{FT}^{48} \le D_{FJ}^{49} \le D_{FM_3 }^{50} \le D_{Fh}^{51} \le\\
\le D_{FM_2 }^{52} \le D_{FM_1 }^{53}
    \le D_{FI}^{54} \le D_{F\Delta }^{55}$.\\
\end{enumerate}

\bigskip
The following equalities hold:
\begin{align}
D_{hM_1 }^8 & = 3D_{hM_2 }^7 = \textstyle{3 \over 2}D_{M_2 M_1 }^4 =\notag\\
& =D_{M_3 h}^{11} =
\textstyle{1 \over 2}D_{M_3 M_1 }^{13} = \textstyle{3 \over 4}D_{M_3 M_2 }^{12}\notag
\intertext{and}
D_{TJ}^{22} & = \textstyle{1 \over 2}D_{TI}^{27} = D_{JI}^{20}.\notag
\end{align}

In view of Lemmas 1 and 2, the measures appearing in the above pyramid
are convex in a pair of probability distributions and can be written as
\begin{equation}
\label{eq6}
D_{AB} : = \sum\limits_{i = 1}^n {q_i } f_{AB} \left( {\frac{p_i }{q_i }}
\right),
\end{equation}

\noindent
where $f_{AB} (x) = f_A (x) - f_B (x)$, $A \ge B$ with the property that
${f}''_{AB} (x) \ge 0$, $\forall x > 0$.

\bigskip
In this paper our aim is to extend the results given by (\ref{eq2}) by taking all possible nonnegative differences given in the above \textbf{pyramid}. These inequalities we have put in nested or sequential form.

\section{Nested Inequalities}

In this section, we shall try to put the measures appearing in above pyramid
in terms of nested or sequential form. This we have done in a theorem below.

\begin{theorem} The following inequalities hold:\\
\[
D_{I\Delta }^{1} \le \textstyle{8 \over 9}D_{M_{1} \Delta }^{3} \le
\textstyle{8 \over {11}}D_{M_{2} \Delta }^{6} \le \textstyle{2 \over
3}D_{h\Delta }^{10} \le \textstyle{8 \over {15}}D_{M_{3} \Delta }^{15} \le
\]
\[
\le \left\{ {\begin{array}{l}
 \left\{ {\begin{array}{l}
 \textstyle{1 \over 2}D_{J\Delta }^{21} \\\\
 \textstyle{8 \over 3}D_{hM_{1} }^{8} \\
 \end{array}} \right\}\le \textstyle{1 \over 3}D_{T\Delta }^{28} \\\\
 \textstyle{8 \over 3}D_{hM_{1} }^{8} \le \textstyle{8 \over 7}D_{M_{3}
I}^{14} \le 2D_{hI}^{9} \\
 \end{array}} \right\}\le \textstyle{8 \over 3}D_{M_{2} I}^{5} \le
\]
\[
\le \left\{ {\begin{array}{l}
 \textstyle{1 \over 3}D_{K_{0} \Delta }^{36} \\\\
 \left\{ {\left\{ {\begin{array}{l}
 D_{TJ}^{22} \\\\
 \textstyle{8 \over {15}}D_{TM_{1} }^{26} \\\\
 \textstyle{8 \over 7}D_{JM_{1} }^{19} \\
 \end{array}} \right\}\le \textstyle{8 \over {13}}D_{TM_{2} }^{25} \le
\textstyle{2 \over 3}D_{Th}^{24} } \right\} \\\\
 8D_{M_{1} I}^{2} \\
 \end{array}} \right.\le
\]
\[
\left. {\begin{array}{l}
 \\
 \\
 \le \left\{ {\begin{array}{l}
 \le \textstyle{2 \over 5}D_{JM_{2} }^{18} \le 2D_{Jh}^{17} \\\\
 \textstyle{8 \over 9}D_{TM_{3} }^{23} \\
 \end{array}} \right. \\
 \\
 \\
 \end{array}} \right\} \le \textstyle{1 \over 2}D_{K_{0} I}^{35} \le
\textstyle{8 \over {15}}D_{K_{0} M_{1} }^{34} \le
\]
\[
\le \textstyle{8 \over {13}}D_{K_{0} M_{2} }^{33} \le \textstyle{2 \over
3}D_{K_{0} h}^{32} \le \left\{ {\begin{array}{l}
 \left\{ {\begin{array}{l}
 D_{K_{0} J}^{30} \\\\
 \textstyle{1 \over 6}D_{\Psi \Delta }^{45} \\
 \end{array}} \right\}\le \frac{1}{5}D_{\Psi I}^{44} \\\\
 \textstyle{8 \over 9}D_{K_{0} M_{3} }^{31} \\
 \end{array}} \right\}\le
\]
\[
\le \textstyle{8 \over {39}}D_{\Psi M_{1} }^{43} \le \textstyle{8 \over
{37}}D_{\Psi M_{2} }^{42} \le \textstyle{2 \over 9}D_{\Psi h}^{41} \le
\left\{ {\begin{array}{l}
 \textstyle{1 \over 4}D_{\Psi J}^{39} \\\\
 \textstyle{8 \over {33}}D_{\Psi M_{3} }^{40} \\
 \end{array}} \right\}\le
\]
\[
\le \textstyle{1 \over 3}D_{\Psi K_{0} }^{37} \le \left\{ {\begin{array}{l}
 \textstyle{1 \over 3}D_{\Psi T}^{38} \\\\
 \textstyle{1 \over 9}D_{F\Delta }^{55} \\
 \end{array}} \right\}\le \textstyle{1 \over 8}D_{FI}^{54} \le \textstyle{8
\over {63}}D_{FM_{1} }^{53} \le
\]
\[
\le \textstyle{8 \over {61}}D_{FM_{2} }^{52} \le \textstyle{2 \over
{15}}D_{Fh}^{51} \le \left\{ {\begin{array}{l}
 \textstyle{1 \over 7}D_{FJ}^{49} \\\\
 \textstyle{8 \over {57}}D_{FM_{3} }^{50} \\
 \end{array}} \right\}\le
\]
\begin{equation}
\label{eq7}
\le \textstyle{1 \over 6}D_{FK_{0} }^{47} \le \textstyle{1 \over
6}D_{FT}^{48} \le \textstyle{1 \over 3}D_{F\Psi }^{46}
\end{equation}
and
\begin{equation}
\label{eq8}
\left\{ {\begin{array}{l}
 \textstyle{1 \over 4}D_{Jh}^{17} \\\\
 D_{M_{1} I}^{2} \\\\
 \textstyle{1 \over 9}D_{TM_{3} }^{23} \\
 \end{array}} \right\}\le D_{JM_{3} }^{16} \le \textstyle{1 \over
{24}}D_{F\Psi }^{46}.
\end{equation}
\end{theorem}

\begin{proof} 
In view of (\ref{eq6}) we shall prove the theorem just
writing the expressions for $f_{AB} (x)$. The rest part is understood
obviously.

\begin{enumerate}
\item \textbf{For }$\bf{D_{I\Delta }^1 (P\vert \vert Q) \le \textstyle{8 \over
9}D_{M_1 \Delta }^3 (P\vert \vert Q)}$\textbf{: }After simplification, we observe
that equivalently, we have to show the following:
\begin{equation}
\label{eq9}
I\le \textstyle{1 \over {36}}\left[ {128M_{1} +\Delta } \right].
\end{equation}
Let us consider the function $g_{I\mathunderscore M_{1} \mathunderscore
\Delta } (x)=\textstyle{{{f}''_{I} (x)} \over {{f}''_{M_{1} \mathunderscore
\Delta } (x)}}$, then we have
\begin{align}
& g_{I\mathunderscore M_{1} \mathunderscore \Delta } (x) =\notag\\
\label{eq10}
& =\frac{\left( {x+1} \right)^{2}\sqrt x \sqrt {2x+2} }{16\left(
{\begin{array}{l}
 \sqrt {2x+2} \left( {2\left( {x+1} \right)^{3}+x^{3/2}} \right)- \\
 -2\left( {x^{3/2}+1} \right)\left( {x+1} \right)^{2} \\
 \end{array}} \right)}.
\end{align}
Here we have
\[
f_{M_{1} \mathunderscore \Delta } (x)=128f_{M_{1} }
(x)+f_{\Delta } (x).
\]

Calculating the first order derivative of the function
$ g_{I\mathunderscore M_{1} \mathunderscore \Delta } (x)$, we get
\begin{align}
& {g}'_{I\mathunderscore M_{1} \mathunderscore \Delta } (x) =\notag\\
& =-\frac{\left( {\begin{array}{l}
 \sqrt {2x+2} \left( {x+1} \right)\left( {\sqrt x -1} \right)\times \\
 \times \left( {x^{2}+x^{3/2}+3x+\sqrt x +1} \right) \\
 \end{array}} \right)\times k_{1} (x)}{16\sqrt x \left( {\begin{array}{l}
 \sqrt {2x+2} \left( {2\left( {x+1} \right)^{3}+x^{3/2}} \right)- \\
 -2\left( {x^{3/2}+1} \right)\left( {x+1} \right)^{2} \\
 \end{array}} \right)^{2}},\notag
\end{align}
where
\[
k_{1} (x)=\sqrt {2x+2} \left( {x^{3/2}+1} \right)-\left( {x+1} \right)^{2}.
\]
This gives
\begin{equation}
\label{eq11}
{g}'_{I\mathunderscore M_{1} \mathunderscore \Delta } (x)\begin{cases}
 {>0,} & {x<1} \\
 {<0,} & {x>1} \\
\end{cases}.
\end{equation}
Expression (\ref{eq11}) is valid only when $k_{1} (x)>0$, $\forall x>0,\;x\ne 1$.
Now, we shall show that $k_{1} (x)>0$, $\forall x>0,\;x\ne 1$. Let us
consider
\[
h_{1} (x)=\left[ {\sqrt {2x+2} \left( {x^{3/2}+1} \right)}
\right]^{2}-\left( {x+1} \right)^{4}.
\]
After simplifications, we have
\[
h_{1} (x)=\left( {x+1} \right)\left( {\sqrt x -1} \right)^{2}\left(
{x^{2}+2x^{3/2}+2\sqrt x +1} \right).
\]
Since $h_{1} (x)>0$, $\forall x>0,\;x\ne 1$ proving that $k_{1} (x)>0$,
$\forall x>0,\;x\ne 1$. Also we have
\begin{align}
\label{eq12}
\beta & = \mathop {\sup }\limits_{x \in (0,\infty )}
g_{I\_M_1 \_\Delta } (x) = \notag\\
& = \mathop {\lim }\limits_{x \to 1} g_{I\_M_1
\_\Delta } (x) = \textstyle{1 \over {36}}.
\end{align}

By the application Lemma 1.2 over (\ref{eq11}) and (\ref{eq12}) we get (\ref{eq9}), proving the
required result.

\bigskip
\textbf{Argument:} \textit{Let }$a$\textit{ and }$b$\textit{ two positive numbers, i.e., }$a > 0$\textit{ and }$b > 0$\textit{. If }$a^2 - b^2 > 0$\textit{, then we can conclude that }$a > b$\textit{ because }$a - b = ({a^2
- b^2)} \mathord{\left/ {\vphantom {{a^2 - b^2)} {(a + b)}}} \right.
\kern-\nulldelimiterspace} {(a + b)}$\textit{. We have used this argument to prove }$k_1 (x) > 0, \forall x > 0, x \ne 1$\textit{. We shall use frequently this argument to prove the other parts of the theorem. }

\begin{remark}\textit{From the above proof we observe that it is sufficient to write expressions similar to (\ref{eq10}), (\ref{eq11}) and (\ref{eq12}). The rest part of the proof follows by the application of Lemma 1.2. In view of this we shall avoid details for the proof of other parts. From now onward, throughout it is understood that }$x > 0,\;x \ne 1.$
\end{remark}

\item \textbf{For }$\bf{D_{M_{1} \Delta }^{3} (P\vert \vert Q)\le \textstyle{9 \over {11}}D_{M_{2} \Delta }^{6} (P\vert \vert Q)}$\textbf{: }Let us consider $g_{M_{1} \Delta \mathunderscore M_{2} \Delta } (x)={{f}''_{M_{1} \Delta } (x)} \mathord{\left/ {\vphantom {{{f}''_{M_{1} \Delta } (x)} {{f}''_{M_{2} \Delta } (x)}}} \right. \kern-\nulldelimiterspace} {{f}''_{M_{2} \Delta } (x)}$, then we have
\begin{align}
& g_{M_{1} \Delta \mathunderscore M_{2} \Delta } (x)  =\notag\\
& \hspace{5pt} =\frac{3\left( {\begin{array}{l}
 \sqrt {2x+2} \left( {\left( {x+1} \right)^{3}-4x^{3/2}} \right)- \\
 -\left( {x^{3/2}+1} \right)\left( {x+1} \right)^{2} \\
 \end{array}} \right)}{\left( {\begin{array}{l}
 2\sqrt {2x+2} \left[ {\left( {x+1} \right)^{3}-6x^{3/2}} \right]- \\
 -\left( {x^{3/2}+1} \right)\left( {x+1} \right)^{2} \\
 \end{array}} \right)}.\notag
\end{align}

This gives
\[
\beta = \lim \limits_{x\to 1} g_{M_{1} \Delta \mathunderscore M_{2} \Delta
} (x)=\textstyle{9 \over {11}}.
\]
Equivalently, we have to show that
\begin{align}
\Omega_{1} & =\textstyle{9 \over {11}}D_{M_{2} \Delta }^{6} -D_{M_{1} \Delta
}^{3} =\notag\\
& =\textstyle{1 \over {22}}\left( {\Delta +24M_{2} -88M_{1} } \right)\ge 0.\notag
\end{align}

We can write $\Omega_{1} :=\sum\nolimits_{i=1}^n q_{i} f_{\Omega_{1} }
\left( {{q_{i} } \mathord{\left/ {\vphantom {{q_{i} } {p_{i} }}} \right.
\kern-\nulldelimiterspace} {p_{i} }} \right)$, where \newline $f_{\Omega_{1} }
(x)=\textstyle{{k_{1} (x)} \over {22(x+1)}}$, with
\begin{align}
k_{2} (x)=& 20x^{3/2}+20\sqrt x +23x^{2}+42x+23\notag\\
& -16\sqrt {2x+2} \left( {\sqrt x +1} \right)\left( {x+1} \right).\notag
\end{align}

Let us consider
\begin{align}
h_{2} (x)= & \left( {20x^{3/2}+20\sqrt x +23x^{2}+42x+23} \right)^{2}\notag\\
& - \left[ {16\sqrt {2x+2} \left( {\sqrt x +1} \right)\left( {x+1} \right)}
\right]^{2}.\notag
\end{align}

After simplifications, we get
\[
h_{2} (x)=\left( {\sqrt x -1} \right)^{6}\left[ {16x+16+\left( {\sqrt x -1}
\right)^{2}} \right].
\]
Since $h_{2} (x)>0$, proving that $k_{2} (x)>0$. This proves the required
result.

\bigskip
\item \textbf{For }$\bf{D_{M_{2} \Delta }^{6} (P\vert \vert Q)\le \textstyle{{11} \over {12}}D_{h\Delta }^{10} (P\vert \vert Q)}$\textbf{: }Let us consider $g_{M_{2} \Delta \mathunderscore h\Delta } (x)={{f}''_{M_{2} \Delta } (x)} \mathord{\left/ {\vphantom {{{f}''_{M_{2} \Delta } (x)} {{f}''_{h\Delta } (x)}}} \right. \kern-\nulldelimiterspace} {{f}''_{h\Delta } (x)}$, then we have
\begin{align}
& g_{M_{2} \Delta \mathunderscore h\Delta } (x)=\notag\\
& =\frac{\left( {\begin{array}{l}
 2\sqrt {2x+2} \left[ {\left( {x+1} \right)^{3}-6x^{3/2}} \right]- \\
 -\left( {x^{3/2}+1} \right)\left( {x+1} \right)^{2} \\
 \end{array}} \right)}{\left( {\begin{array}{l}
 3\sqrt {2x+2} \left( {\sqrt x -1} \right)^{2}\times \\
 \times \left( {x^{2}+2x^{3/2}+6x+2\sqrt x +1} \right) \\
 \end{array}} \right)}.\notag
\end{align}
This gives
\[
\beta = \lim\limits_{x\to 1} g_{M_{2} \Delta \mathunderscore h\Delta }
(x)=\textstyle{{11} \over {12}}.
\]
Equivalently, we have to show that
\begin{align}
\Omega_{2} & =\textstyle{{11} \over {12}}D_{h\Delta }^{10} -D_{M_{2} \Delta
}^{6} =\notag\\
& =\textstyle{1 \over {48}}\left( {44h+\Delta -64M_{2} } \right)\ge 0.\notag
\end{align}

We can write $\Omega_{2} :=\sum\nolimits_{i=1}^n q_{i} f_{\Omega_{2} }
\left( {{q_{i} } \mathord{\left/ {\vphantom {{q_{i} } {p_{i} }}} \right.
\kern-\nulldelimiterspace} {p_{i} }} \right)$, where \newline $f_{\Omega_{2} }
(x)=\textstyle{{k_{3} (x)} \over {48(x+1)}}$, with $k_{3} (x)=k_{2} (x)>0$.
This proves the required result.

\bigskip
\item \textbf{For }$\bf{D_{h\Delta }^{10} (P\vert \vert Q)\le \textstyle{4 \over 5}D_{M_{3} \Delta }^{15} (P\vert \vert Q)}$\textbf{: }Let us consider $g_{h\Delta \mathunderscore M_{3} \Delta } (x)={{f}''_{h\Delta } (x)} \mathord{\left/ {\vphantom {{{f}''_{h\Delta } (x)} {{f}''_{M_{3} \Delta } (x)}}} \right. \kern-\nulldelimiterspace} {{f}''_{M_{3} \Delta } (x)}$, then we have
\begin{align}
& g_{h\Delta \mathunderscore M_{3} \Delta } (x)=\notag\\
& =\frac{\left( {\begin{array}{l}
 \left( {x^{2}+2x^{3/2}+6x+2\sqrt x +1} \right)\times \\
 \times \left( {\sqrt x -1} \right)^{2}\sqrt {2x+2} \\
 \end{array}} \right)}{2\left[ {\left( {x^{3/2}+1} \right)\left( {x+1}
\right)^{2}-4x^{3/2}\sqrt {2x+2} } \right]}.\notag
\end{align}

This gives
\[
\beta = \lim\limits_{x\to 1} g_{h\Delta \mathunderscore M_{3} \Delta }
(x)=\textstyle{4 \over 5}.
\]
Equivalently, we have to show that
\begin{align}
\Omega_{3} & =\textstyle{4 \over 5}D_{M_{3} \Delta }^{15} -D_{h\Delta }^{10}
=\notag\\
& =\textstyle{1 \over {20}}\left( {64M_{3} +\Delta -20h} \right)\ge 0.\notag
\end{align}

We can write $\Omega_{3} :=\sum\nolimits_{i=1}^n q_{i} f_{\Omega_{3} }
\left( {{q_{i} } \mathord{\left/ {\vphantom {{q_{i} } {p_{i} }}} \right.
\kern-\nulldelimiterspace} {p_{i} }} \right)$, where \newline $f_{\Omega_{3} }
(x)=\textstyle{{k_{4} (x)} \over {20(x+1)}}$, with $k_{4} (x)=k_{2} (x)>0$.
This proves the required result.

\bigskip
\item \textbf{For }$\bf{D_{M_{3} \Delta }^{15} (P\vert \vert Q)\le 5D_{hM_{1} }^{8} (P\vert \vert Q)}$\textbf{: }Let us consider $g_{M_{3} \Delta \mathunderscore hM_{1} } (x)={{f}''_{M_{3} \Delta } (x)} \mathord{\left/ {\vphantom {{{f}''_{M_{3} \Delta } (x)} {{f}''_{hM_{1} } (x)}}} \right. \kern-\nulldelimiterspace} {{f}''_{hM_{1} } (x)}$, then we have
\begin{align}
& g_{M_{3} \Delta \mathunderscore hM_{1} } (x)=\notag\\
& =\frac{2\left[ {\left( {x^{3/2}+1} \right)\left( {x+1}
\right)^{2}-4x^{3/2}\sqrt {2x+2} } \right]}{\left( {x+1} \right)^{2}\left[
{2x^{3/2}+2-\left( {x+1} \right)\sqrt {2x+2} } \right]}.\notag
\end{align}

This gives
\[
\beta = \lim\limits_{x\to 1} g_{M_{3} \Delta \mathunderscore hM_{1} }
(x)=5.
\]
Equivalently, we have to show that
\begin{align}
\Omega_{4} & =\textstyle{4 \over 5}D_{M_{3} \Delta }^{15} -D_{h\Delta }^{10} =\notag\\
& =\textstyle{1 \over 4}\left( {20h+\Delta -80M_{1} -64M_{3} } \right)\ge 0.\notag
\end{align}

We can write $\Omega_{4} :=\sum\nolimits_{i=1}^n q_{i} f_{\Omega_{4} }
\left( {{q_{i} } \mathord{\left/ {\vphantom {{q_{i} } {p_{i} }}} \right.
\kern-\nulldelimiterspace} {p_{i} }} \right)$, where \newline $f_{\Omega_{4} }
(x)=\textstyle{{k_{5} (x)} \over {4(x+1)}}$, with $k_{5} (x)=k_{2} (x)>0$.
This proves the required result.

\bigskip
\item \textbf{For }$\bf{D_{M_3 \Delta }^{15} (P\vert \vert Q) \le \textstyle{{15}
\over {16}}D_{J\Delta }^{21} (P\vert \vert Q)}$\textbf{: }Let us consider $g_{M_{3} \Delta \mathunderscore J\Delta } (x)={{f}''_{M_{3} \Delta } (x)} \mathord{\left/ {\vphantom {{{f}''_{M_{3} \Delta } (x)} {{f}''_{J\Delta } (x)}}} \right. \kern-\nulldelimiterspace} {{f}''_{J\Delta } (x)}$, then we have
\[
g_{M_{3} \Delta \mathunderscore J\Delta } (x)=
\frac{4\sqrt x \left( {\begin{array}{l}
 \left( {x^{(3/2)}+1} \right)\left( {x+1} \right)^{2}- \\
 -4x^{3/2}\sqrt {2x+2} \\
 \end{array}} \right)}{\sqrt {2x+2} \left( {x^{2}+6x+1} \right)\left( {x-1}
\right)^{2}}
\]
and
\begin{align}
& {g}'_{M_{3} \Delta \mathunderscore J\Delta }(x) =\notag\\
& =-\frac{2\left( {x+1} \right)\times k_{6} (x)}{\sqrt {2x+2} \sqrt x \left(
{x^{2}+6x+1} \right)^{2}\left( {x-1} \right)^{3}},\notag
\end{align}

where
\begin{align}
& k_{6} (x) = -16x^{3/2}\sqrt {2x+2} \left( {x+1} \right)^{2}+\notag\\
& +\left( {\sqrt x +1} \right)\left( {\begin{array}{l}
 \left( {x^{4}+9x^{2}+1} \right)\left( {\sqrt x -1} \right)^{2}+ \\
 +x^{9/2}+x^{4}+2x^{7/2}+28x^{3}+ \\
 +28x^{2}+2x^{3/2}+x+\sqrt x \\
 \end{array}} \right).\notag
\end{align}

This gives
\[
{g}'_{M_{3} \Delta \mathunderscore J\Delta } (x)\begin{cases}
 {>0,} & {x<1} \\
 {<0,} & {x>1} \\
\end{cases},
\]
provided $k_{6} (x)>0$. Now we shall show that $k_{6} (x)>0$. Let us
consider
\begin{align}
& h_{6} (x)=-\left[ {16x^{3/2}\sqrt {2x+2} \left( {x+1} \right)^{2}}
\right]^{2}+\notag\\
& +\left[ {\left( {\sqrt x +1} \right)\left( {\begin{array}{l}
 \left( {x^{4}+9x^{2}+1} \right)\left( {\sqrt x -1} \right)^{2}+ \\
 +x^{9/2}+x^{4}+2x^{7/2}+28x^{3}+ \\
 +28x^{2}+2x^{3/2}+x+\sqrt x \\
 \end{array}} \right)} \right]^{2}.\notag
\end{align}

Simplifying the above expression, we get
\begin{align}
h_{6} (x)= & \left( {\sqrt x -1} \right)^{4}\times\notag\\
& \times \left( {\begin{array}{l}
 1+4\sqrt x +538x^{5/2}+36x^{3/2}+ \\
 +1460x^{7/2}+2196x^{9/2}+538x^{13/2}+ \\
 +1537x^{5}+1537x^{4}+908x^{3}+12x^{8} \\
 +166x^{2}+12x+4x^{17/2}+36x^{15/2}+ \\
 +1460x^{11/2}+x^{9}+908x^{6}+166x^{7} \\
 \end{array}} \right).\notag
\end{align}

Since $h_{6} (x)>0$, proving that $k_{6} (x)>0$. Also we have
\[
\beta =\mathop {\sup }\limits_{x\in
(0,\infty )} g_{M_{3} \Delta \mathunderscore J\Delta } (x)
=\lim\limits_{x\to 1} g_{M_{3} \Delta \mathunderscore J\Delta }
(x)=\textstyle{{15} \over {16}}.
\]

\item \textbf{For }$\bf{D_{J\Delta }^{21} (P\vert \vert Q) \le \textstyle{2 \over
3}D_{T\Delta }^{28} (P\vert \vert Q)}$\textbf{: }It holds in view of (\ref{eq2}).

\bigskip
\item \textbf{For }$\bf{D_{hM_{1} }^{8} (P\vert \vert Q)\le \textstyle{1 \over 8}D_{T\Delta }^{28} (P\vert \vert Q)}$\textbf{: }Let us consider $g_{hM_{1} \mathunderscore T\Delta } (x)={{f}''_{hM_{1} } (x)} \mathord{\left/ {\vphantom {{{f}''_{hM_{1} } (x)} {{f}''_{T\Delta } (x)}}} \right. \kern-\nulldelimiterspace} {{f}''_{T\Delta } (x)}$, then we have
\begin{align}
& g_{hM_{1} \mathunderscore T\Delta } (x)=\notag\\
& =\frac{\sqrt x \left( {x+1} \right)^{2}\left[ {2\left( {x^{3/2}+1}
\right)-\sqrt {2x+2} \left( {x+1} \right)} \right]}{\sqrt {2x+2} \left(
{x^{2}+4x+1} \right)\left( {x-1} \right)^{2}}\notag
\intertext{and}
& {g}'_{hM_{1} \mathunderscore T\Delta } (x)=-\frac{\left( {x+1} \right)k_{7}
(x)}{\left( {\begin{array}{l}
 2\left( {x^{2}+4x+1} \right)^{2}\times \\
 \times \left( {x-1} \right)^{3}\sqrt x \sqrt {2x+2} \\
 \end{array}} \right)},\notag
\intertext{where}
& k_{7} (x)=\notag\\
& =2\left( {\sqrt x +1} \right)\left( {\begin{array}{l}
 \left( {x^{4}+5x^{2}+1} \right)\left( {\sqrt x -1} \right)^{2}+ \\
 +3x^{4}+x^{9/2}+20x^{3}+ \\
 +20x^{2}+3x+\sqrt x \\
 \end{array}} \right)-\notag\\
& -\sqrt {2x+2} \left( {x^{4}+6x^{3}+34x^{2}+6x+1} \right)\left( {x+1}
\right).\notag
\end{align}

This give
\[
{g}'_{hM_{1} \mathunderscore T\Delta } (x)\begin{cases}
 {>0,} & {x<1} \\
 {<0,} & {x>1} \\
\end{cases},
\]
provided $k_{7} (x)>0$. Now we shall show that $k_{7} (x)>0$. Let us
consider
\begin{align}
& h_{7} (x) =\notag\\
& = \left[ {2\left( {\sqrt x +1} \right)\left( {\begin{array}{l}
 \left( {x^{4}+5x^{2}+1} \right)\left( {\sqrt x -1} \right)^{2}+ \\
 +3x^{4}+x^{9/2}+20x^{3}+ \\
 +20x^{2}+3x+\sqrt x \\
 \end{array}} \right)} \right]^{2}\notag\\
& \hspace{10pt} -\left[ {\sqrt {2x+2} \left( {x+1} \right)\left( {\begin{array}{l}
 x^{4}+6x^{3}+ \\
 +34x^{2}+6x+1 \\
 \end{array}} \right)} \right]^{2}.\notag
\end{align}

After simplifications, we get
\begin{align}
& h_{7} (x)=2\left( {\sqrt x -1} \right)^{4}\times\notag\\
& \times \left( {\begin{array}{l}
 x^{9}+4x^{17/2}+7x^{8}+24x^{15/2}+44x^{7}+ \\
 +164x^{13/2}+104x^{6}+222x^{11/2}+ \\
 +42x^{7/2}\left( {x+1} \right)\left( {\sqrt x -1} \right)^{2}+300x^{9/2}+
\\
 +222x^{7/2}+104x^{3}+164x^{5/2}+ \\
 +44x^{2}+24x^{3/2}+7x+4\sqrt x +1 \\
 \end{array}} \right).\notag
\end{align}

Since $h_{7} (x)>0$, this gives that $k_{7} (x)>0$. Also we have
\[
\beta =\mathop {\sup }\limits_{x\in (0,\infty )}
g_{hM_{1} \mathunderscore T\Delta } (x)
=\lim\limits_{x\to 1} g_{hM_{1} \mathunderscore T\Delta }
(x)=\textstyle{1 \over 8}.
\]

\item \textbf{For $\bf{D_{M_{3} h}^{11} (P\vert \vert Q)\le \textstyle{3 \over 7}D_{M_{3} I}^{14} (P\vert \vert Q)}$:} Let us consider $g_{M_{3} h\mathunderscore M_{3} I} (x)={{f}''_{M_{3} h} (x)} \mathord{\left/ {\vphantom {{{f}''_{M_{3} h} (x)} {{f}''_{M_{3} I} (x)}}} \right. \kern-\nulldelimiterspace} {{f}''_{M_{3} I} (x)}$, then we have
\[
g_{M_{3} h\mathunderscore M_{3} I} (x)=\frac{2(x^{3/2}+1)-(x+1)\sqrt {2x+2}
}{2\left[ {(x^{3/2}+1)-\sqrt x \sqrt {2x+2} } \right]}
\]
and
\begin{align}
& {g}'_{M_{3} h\mathunderscore M_{3} I} (x)=\notag\\
& =-\frac{\left( {\sqrt x -1} \right)\times k_{8} (x)}{2\sqrt x \sqrt {2x+2}
\left[ {\left( {x^{3/2}+1} \right)-\sqrt x \sqrt {2x+2} } \right]^{2}},\notag
\end{align}

where
\begin{align}
& k_{8} (x)= -\left( {x+1} \right)\left( {\sqrt x +1} \right)\sqrt {2x+2}\,+\notag\\
& +\left( {x+1} \right)\left( {\sqrt x -1} \right)^{2}+\left( {x^{3/2}+1} \right)\left( {\sqrt x +1} \right)+4x. \notag
\end{align}

This gives
\[
{g}'_{M_{3} h\mathunderscore M_{3} I} (x)\begin{cases}
 {>0,} & {x<1} \\
 {<0,} & {x>1} \\
\end{cases},
\]
provided $k_{8} (x)>0$. In order to prove $k_{8} (x)>0$, let us consider
\begin{align}
h_{8} (x)=& \left( {\begin{array}{l}
 \left( {x+1} \right)\left( {\sqrt x -1} \right)^{2}+ \\
 +\left( {x^{3/2}+1} \right)\left( {\sqrt x +1} \right)+4x \\
 \end{array}} \right)^{2} - \notag\\
& -\left[ {\left( {x+1} \right)\left( {\sqrt x +1} \right)\sqrt {2x+2} }
\right]^{2}.\notag
\end{align}

After simplifications, we have
\[
h_{8} (x)=\left( {2x+1} \right)\left( {x+2} \right)\left( {\sqrt x -1}
\right)^{4}.
\]
Since $h_{8} (x)>0$, this gives that $k_{8} (x)>0$. Also we have
\[
\beta =\mathop {\sup }\limits_{x\in (0,\infty )}
g_{M_{3} h\mathunderscore M_{3} I} (x)
=\lim\limits_{x\to 1} g_{M_{3} h\mathunderscore M_{3} I} (x)=\textstyle{3 \over 7}.
\]

\item \textbf{For $\bf{D_{M_{3} I}^{14} (P\vert \vert Q)\le \textstyle{7 \over 4}D_{hI}^{9} (P\vert \vert Q)}$:} Let us consider \\ $g_{M_{3} I\mathunderscore hI} (x)={{f}''_{M_{3} I} (x)} \mathord{\left/ {\vphantom {{{f}''_{M_{3} I} (x)} {{f}''_{M_{3} h} (x)}}} \right. \kern-\nulldelimiterspace} {{f}''_{M_{3} h} (x)}$, then we have
\[
g_{M_{3} I\mathunderscore hI} (x)=\frac{2\left[ {x^{3/2}+1-\sqrt x \sqrt
{2x+2} } \right]}{\left( {\sqrt x -1} \right)^{2}\sqrt {2x+2} }
\]
and
\[
{g}'_{M_{3} I\mathunderscore hI} (x)=-\frac{k_{9} (x)}{\left( {\sqrt x -1}
\right)^{3}\left( {x+1} \right)\sqrt x \sqrt {2x+2} },
\]
where $k_{9} (x)=k_{8} (x)>0$. This give
\[
{g}'_{M_{3} I\mathunderscore hI} (x)\begin{cases}
 {>0,} & {x<1} \\
 {<0,} & {x>1} \\
\end{cases}.
\]
Also we have
\[
\beta=\mathop {\sup }\limits_{x\in (0,\infty )}
g_{M_{3} I\mathunderscore hI} (x)
=\lim\limits_{x\to 1} g_{M_{3} I\mathunderscore hI} (x)=\textstyle{7 \over 4}.
\]

\item \textbf{For }$\bf{D_{hI}^{9} (P\vert \vert Q)\le \textstyle{4 \over 3}D_{M_{2} I}^{5} (P\vert \vert Q)}$\textbf{: }Let us consider \\ $g_{hI\mathunderscore M_{2} I} (x)={{f}''_{hI} (x)} \mathord{\left/ {\vphantom {{{f}''_{hI} (x)} {{f}''_{M_{2} I} (x)}}} \right. \kern-\nulldelimiterspace} {{f}''_{M_{2} I} (x)}$, then we have
\begin{align}
& g_{hI\mathunderscore M_{2} I} (x)=\notag\\
& =\frac{3\left( {\sqrt x -1} \right)^{2}\sqrt {2x+2} }{2\left[ {\sqrt {2x+2}
\left( {2x-3\sqrt x +2} \right)-\left( {x^{3/2}+1} \right)} \right]}\notag
\end{align}
and
\[
{g}'_{hI\mathunderscore M_{2} I} (x)=\frac{3\left( {1- \sqrt x}
\right)\times k_{10} (x)}{\left( {\begin{array}{l}
 2\sqrt x \sqrt {2x+2} \left[ {-\left( {x^{3/2}+1} \right)+} \right. \\
 \left. {+\sqrt {2x+2} \left( {2x-3\sqrt x +2} \right)} \right]^{2} \\
 \end{array}} \right)},
\]
where $k_{10} (x)=k_{8} (x)>0$. This gives
\[
{g}'_{hI\mathunderscore M_{2} I} (x)\begin{cases}
 {>0,} & {x<1} \\
 {<0,} & {x>1} \\
\end{cases}.
\]
Also we have
\[
\beta =\mathop {\sup }\limits_{x\in (0,\infty )}
g_{hI\mathunderscore M_{2} I} (x)
=\lim\limits_{x\to 1} g_{hI\mathunderscore M_{2} I} (x)=\textstyle{4
\over 3}.
\]

\item \textbf{For }$\bf{D_{M_{2} I}^{5} (P\vert \vert Q)\le \textstyle{1 \over 8}D_{K_{0} \Delta}^{36} (P\vert \vert Q)}$\textbf{: }Let us consider $g_{M_{2} I\mathunderscore K_{0} \Delta } (x)=\textstyle{{{f}''_{M_{2} I} (x)} \over {{f}''_{K_{0} \Delta } (x)}}$, then we have
\begin{align}
& g_{M_{2} I\mathunderscore K_{0} \Delta } (x)=\frac{16x\left( {x+1}
\right)^{2}}{3\sqrt {2x+2} \left( {\sqrt x -1} \right)^{2}}\times\notag\\
& \times \frac{\left[ {\sqrt {2x+2} \left( {2x-3\sqrt x +2} \right)-\left(
{x^{3/2}+1} \right)} \right]}{\left( {\begin{array}{l}
 3x^{4}+6x^{7/2}+20x^{3}+34x^{5/2} \\
 +66x^{2}+34x^{3/2}+20x+6\sqrt x +3 \\
 \end{array}} \right)}\notag
\end{align}
and
\begin{align}
& {g}'_{M_{2} I\mathunderscore K_{0} \Delta } (x) = -\frac{8(x+1)\times k_{11}
(x)}{\sqrt {2x+2} \left( {\sqrt x -1} \right)^{3}}\times\notag\\
& \times \frac{1}{\left( {\begin{array}{l}
 3x^{4}+6x^{7/2}+20x^{3}+34x^{5/2}+ \\
 +66x^{2}+34x^{3/2}+20x+6\sqrt x +3 \\
 \end{array}} \right)^{2}},\notag
\end{align}
where
\begin{align}
& k_{11} (x)= \sqrt {2x+2} \left( {\sqrt x +1} \right)u(x)-\notag\\
& -\left( {\begin{array}{l}
 2x^{7}+2x^{13/2}+7x^{6}+12x^{11/2}+ \\
 +11x^{5}+94x^{9/2}+76x^{4}+ \\
 +104x^{7/2}+76x^{3}+94x^{5/2}+11x^{2} \\
 +12x^{3/2}+7x+2\sqrt x +2 \\
 \end{array}} \right),\notag
\end{align}
with
\[
u(x)=\left( {\begin{array}{l}
 4x^{6}-9x^{11/2}+24x^{5}-41x^{9/2}+ \\
 +60x^{4}+50x^{7/2}-48x^{3}+50x^{5/2}+ \\
 +60x^{2}-41x^{3/2}+24x-9\sqrt x +4 \\
 \end{array}} \right).
\]
This gives
\[
{g}'_{M_{2} I\mathunderscore K_{0} \Delta } (x)\begin{cases}
 {>0,} & {x<1} \\
 {<0,} & {x>1} \\
\end{cases},
\]
provided $k_{11} (x)>0$. In order to prove $k_{11} (x)>0$, let us consider
\[
v(t)=u(t^{2})=\left( {\begin{array}{l}
 4t^{12}-9t^{11}+24t^{10}-41t^{9}+ \\
 +60t^{8}+50t^{7}-48t^{6}+50t^{5}+ \\
 +60t^{4}-41t^{3}+24t^{2}-9t+4 \\
 \end{array}} \right).
\]
Solving the polynomial equation $v(t)=0$, we observe that there are no real
solutions. All the twelve solutions are complex and are given by
\begin{center}
$-0.9437538663\pm 0.3306488166\mbox{\thinspace }I,$\\
$-0.3823946004\pm 2.215272138\mbox{\thinspace }I,$\\
$-0.07566691909\pm 0.4383503779\mbox{\thinspace }I,$\\
$\hspace{3pt} 0.3872722043\pm 0.2946782782\mbox{\thinspace }I,$\\
$\hspace{3pt} 0.5042070498\pm 0.8635827991\mbox{\thinspace }I,$\\
$1.635336132\pm 1.244339331\mbox{\thinspace }I.$
\end{center}

This means that for all $t>0$, either $v(t)>0$ or $v(t)<0$. Calculating a
particular value of $v(t)$, for example for $t=1$, we get $v(1)=128>0$. This
means that $v(t)>0$, for all $t>0$, and hence $u(x)>0$, $\forall x>0$. Let
us consider
\begin{align}
& h_{11} (x)=\left[ {\sqrt {2x+2} \left( {\sqrt x +1} \right) u(x)}
\right]^{2} - \left[\sqrt x \left( {x+1} \right)\right]^{2} \times\notag\\
& \times {\left( {\begin{array}{l}
 2x^{7}+2x^{13/2}+7x^{6}+12x^{11/2}+ \\
 +11x^{5}+94x^{9/2}+76x^{4}+ \\
 +104x^{7/2}+76x^{3}+94x^{5/2}+ \\
 +11x^{2}+12x^{3/2}+7x+2\sqrt x +2 \\
 \end{array}} \right)}^{2}.\notag
\end{align}

After simplifications, we have
\begin{align}
& h_{11} (x)=x\left( {x+1} \right)^{2}\left( {\sqrt x -1} \right)^{4}\times\notag\\
& \times \left( {\begin{array}{l}
 8+2608x^{5}+218x+112x^{3/2}+451x^{2}+ \\
 +24\sqrt x +1910x^{3}+1180x^{5/2}+5612x^{7/2}+ \\
 +1420x^{9/2}+1124x^{11/2}+1124x^{13/2}+ \\
 +1910x^{9}+3745x^{8}+3745x^{4}+218x^{11}+ \\
 +1983x^{4}\left( {x^{2}-1} \right)^{2}+1420x^{15/2}+ \\
 +5612x^{17/2}+2608x^{7}+451x^{10}+ \\
 +28x^{12}+1180x^{19/2}+24x^{23/2}+112x^{21/2} \\
 \end{array}} \right).\notag
\end{align}
Since $h_{11} (x)>0$, this gives that $k_{11} (x)>0$. Also we have
\[
\beta =\mathop {\sup }\limits_{x\in
(0,\infty )} g_{M_{2} I\mathunderscore K_{0} \Delta } (x)=
\lim\limits_{x\to 1} g_{M_{2} I\mathunderscore K_{0} \Delta }
(x)=\textstyle{1 \over 8}.
\]

\item \textbf{For }$\bf{D_{M_{2} I}^{5} (P\vert \vert Q)\le \textstyle{3 \over 8}D_{TJ}^{22} (P\vert \vert Q)}$\textbf{: }Let us consider $g_{M_{2} I\mathunderscore TJ} (x)={{f}''_{M_{2} I} (x)} \mathord{\left/ {\vphantom {{{f}''_{M_{2} I} (x)} {{f}''_{TJ} (x)}}} \right. \kern-\nulldelimiterspace} {{f}''_{TJ} (x)}$, then we have
\begin{align}
& g_{M_{2} I\mathunderscore TJ} (x)=\notag\\
& =\frac{4\left[ {\left( {2x-3\sqrt x +2} \right)\sqrt {2x+2} -\left(
{x^{3/2}+1} \right)} \right]}{3\sqrt {2x+2} \left( {x-1} \right)^{2}}\notag
\end{align}
and
\[
{g}'_{M_{2} I\mathunderscore TJ} (x)=-\frac{2\times k_{12} (x)}{3\sqrt
{2x+2} \sqrt x \left( {x+1} \right)\left( {x-1} \right)^{3}},
\]
where
\begin{align}
k_{12} (x)=& 2\sqrt {2x+2} \times\notag\\
& \times \left( {\begin{array}{l}
 x^{2}\left( {\sqrt x -2} \right)^{2}+3x\left( {\sqrt x -1} \right)^{2}+ \\
 +\left( {2\sqrt x -1} \right)^{2}+\sqrt x \left( {x^{2}+1} \right) \\
 \end{array}} \right)-\notag\\
& -\left( {x^{7/2}+3x^{5/2}+4x^{3/2}+3x+4x^{2}+1} \right).\notag
\end{align}
This gives
\[
{g}'_{M_{2} I\mathunderscore TJ} (x)\begin{cases}
 {>0} & {x<1} \\
 {<0} & {x>1} \\
\end{cases},
\]
provided $k_{12} (x)>0$. In order to prove $k_{12} (x)>0$, let us
consider
\begin{align}
& h_{12} (x)=\left( {2\sqrt {2x+2} } \right)^{2}\times\notag\\
& \times \left( {\begin{array}{l}
 x^{2}\left( {\sqrt x -2} \right)^{2}+3x\left( {\sqrt x -1} \right)^{2}+ \\
 +\left( {2\sqrt x -1} \right)^{2}+\sqrt x \left( {x^{2}+1} \right) \\
 \end{array}} \right)^{2}\notag\\
& -\left( {x^{7/2}+3x^{5/2}+4x^{3/2}+3x+4x^{2}+1} \right)^{2}.\notag
\end{align}

After simplifications, we have
\begin{align}
& h_{12} (x)=\left( {\sqrt x -1} \right)^{4}\times\notag\\
& \times \left( {\begin{array}{l}
 \left( {x+1} \right)\left( {2x^{4}+45x^{2}+2} \right)+ \\
 +5x^{4}\left( {\sqrt x -2} \right)^{2}+5\left( {2\sqrt x -1} \right)^{2}+
\\
 +x\left( {\sqrt x -1} \right)^{2}\left( {42x^{2}+65x+42} \right) \\
 \end{array}} \right). \notag
\end{align}

Since $h_{12} (x)>0$, this gives that $k_{12} (x)>0$. Also we have
\[
\beta =\mathop {\sup }\limits_{x\in (0,\infty )}
g_{M_{2} I\mathunderscore TJ} (x)=
\lim\limits_{x\to 1} g_{M_{2} I\mathunderscore TJ} (x)=\textstyle{3
\over 8}.
\]

\item \textbf{For }$\bf{D_{M_{2} I}^{5} (P\vert \vert Q)\le \textstyle{1 \over 5}D_{TM_{1} }^{26} (P\vert \vert Q)}$\textbf{: }Let us consider $g_{M_{2} I\mathunderscore TM_{1} } (x)={{f}''_{M_{2} I} (x)} \mathord{\left/ {\vphantom {{{f}''_{M_{2} I} (x)} {{f}''_{TM_{1} } (x)}}} \right. \kern-\nulldelimiterspace} {{f}''_{TM_{1} } (x)}$, then we have
\begin{align}
& g_{M_{2} I\mathunderscore TM_{1} } (x)=\notag\\
& =\frac{2\sqrt x \left( {\begin{array}{l}
 \left( {2x-3\sqrt x +2} \right)\sqrt {2x+2} - \\
 -\left( {x^{3/2}+1} \right) \\
 \end{array}} \right)}{3\left( {\begin{array}{l}
 \sqrt {2x+2} \left( {x^{2}-2x^{3/2}-2\sqrt x +1} \right)+ \\
 +2\sqrt x \left( {x^{3/2}+1} \right) \\
 \end{array}} \right)}\notag
\end{align}
and
\begin{align}
& {g}'_{M_{2} I\mathunderscore TM_{1} } (x)=\notag\\
& =-\frac{\left( {\sqrt x -1} \right)\times k_{13} (x)}{\left(
{\begin{array}{l}
 3\sqrt x \left( {x+1} \right)\left[ {-2\sqrt x \left( {x^{3/2}+1} \right)+}
\right. \\
 \left. {+\sqrt {2x+2} \left( {x^{2}-2x^{3/2}-2\sqrt x +1} \right)}
\right]^{2} \\
 \end{array}} \right)},\notag
\end{align}
where
\begin{align}
& k_{13} (x)=4\left( {x+1} \right)^{2}\times\notag\\
& \times \left[ {\left( {x^{3/2}+1} \right)\left( {\sqrt x -1}
\right)^{2}+3x\left( {\sqrt x +1} \right)} \right]-\notag\\
& -\sqrt {2x+2} \left( {\begin{array}{l}
 x^{4}+x^{7/2}+7x^{3}+5x^{5/2}+ \\
 +20x^{2}+5x^{3/2}+7x+\sqrt x +1 \\
 \end{array}} \right).\notag
\end{align}
This gives
\[
{g}'_{M_{2} I\mathunderscore TM_{1} } (x)\begin{cases}
 {>0,} & {x<1} \\
 {<0,} & {x>1} \\
\end{cases},
\]
provided $k_{13} (x)>0$. In order to prove $k_{13} (x)>0$, let us consider
\begin{align}
& h_{13} (x)=\left[ {4\left( {x+1} \right)^{2}} \right]^{2}\times\notag\\
& \times \left[ {\left( {x^{3/2}+1} \right)\left( {\sqrt x -1}
\right)^{2}+3x\left( {\sqrt x +1} \right)} \right]^{2}-\notag\\
& -\left\{ {\sqrt {2x+2} \left( {\begin{array}{l}
 x^{4}+x^{7/2}+7x^{3}+5x^{5/2}+ \\
 +20x^{2}+5x^{3/2}+7x+\sqrt x +1 \\
 \end{array}} \right)} \right\}^{2}.\notag
\end{align}
After simplifications, we have
\begin{align}
& h_{13} (x)=2\left( {x+1} \right)\left( {\sqrt x -1} \right)^{4}\times\notag\\
& \times \left( {\begin{array}{l}
 6x^{6}+x^{5}\left( {\sqrt x -3} \right)^{2}+30x^{5}+36x^{9/2}+ \\
 +60x^{4}+114x^{7/2}+40x^{3}+114x^{5/2}+ \\
 +60x^{2}+36x^{3/2}+30+\left( {3\sqrt x -1} \right)^{2}+6 \\
 \end{array}} \right).\notag
\end{align}

Since $h_{13} (x)>0$, this gives that $k_{13} (x)>0$. Also we have
\[
\beta =\mathop {\sup }\limits_{x\in (0,\infty )}
g_{M_{2} I\mathunderscore TM_{1} } (x)=
\lim\limits_{x\to 1} g_{M_{2} I\mathunderscore TM_{1} } (x)=\textstyle{1
\over 5}.
\]

\item \textbf{For }$\bf{D_{M_{2} I}^{5} (P\vert \vert Q)\le \textstyle{3 \over 7}D_{JM_{1} }^{19} (P\vert \vert Q)}$\textbf{: }Let us consider $g_{M_{2} I\mathunderscore JM_{1} } (x)={{f}''_{M_{2} I} (x)} \mathord{\left/ {\vphantom {{{f}''_{M_{2} I} (x)} {{f}''_{JM_{1} } (x)}}} \right. \kern-\nulldelimiterspace} {{f}''_{JM_{1} } (x)}$, then we have
\begin{align}
& g_{M_{2} I\mathunderscore JM_{1} } (x)=\notag\\
& = \frac{4x\left( {\begin{array}{l}
 \left( {2x-3\sqrt x +2} \right)\times \\
 \times \sqrt {2x+2} -\left( {x^{3/2}+1} \right) \\
 \end{array}} \right)}{3\left( {\begin{array}{l}
 \sqrt x \sqrt {2x+2} \left( {x+1} \right)\times \\
 \times \left( {x-4x+1} \right)+4x\left( {x^{3/2}+1} \right) \\
 \end{array}} \right)}\notag
\end{align}
and
\begin{align}
& {g}'_{M_{2} I\mathunderscore JM_{1} } (x)=\notag\\
& =-\frac{2\sqrt x \left( {\sqrt x -1} \right)\times k_{14} (x)}{\left(
{\begin{array}{l}
 3\left( {x+1} \right)\left[ {4x\left( {x^{(3/2)}+1} \right)+\left( {x+1}
\right)\times } \right. \\
 \left. {\times \sqrt x \sqrt {2x+2} \left( {x-4\sqrt x +1} \right)}
\right]^{2} \\
 \end{array}} \right)},\notag
\end{align}
where
\begin{align}
& k_{14} (x)=2\sqrt {2x+2} \left( {\sqrt x +1} \right)\left( {x+1}
\right)\times\notag\\
& \times \left[ {x\left( {\sqrt x -2} \right)^{2}+\left( {2\sqrt x -1}
\right)^{2}+\sqrt x \left( {x+1} \right)} \right]-\notag\\
& -\left( {\begin{array}{l}
 x^{4}+x^{7/2}+11x^{3}+3x^{5/2}+ \\
 +32x^{2}+3x^{3/2}+11x+\sqrt x +1 \\
 \end{array}} \right).\notag
\end{align}

This gives
\[
{g}'_{M_{2} I\mathunderscore JM_{1} } (x)\begin{cases}
 {>0} & {x<1} \\
 {<0} & {x>1} \\
\end{cases},
\]
provided $k_{14} (x)>0$. In order to prove $k_{14} (x)>0$, let us consider
\begin{align}
& h_{14} (x)=\left\{ {\begin{array}{l}
 2\sqrt {2x+2} \left( {\sqrt x +1} \right)\left( {x+1} \right)\times \\
 \times \left[ {\begin{array}{l}
 x\left( {\sqrt x -2} \right)^{2}+\left( {2\sqrt x -1} \right)^{2} \\
 +\sqrt x \left( {x+1} \right) \\
 \end{array}} \right] \\
 \end{array}} \right\}^{2}-\notag\\
& -\left( {\begin{array}{l}
 x^{4}+x^{7/2}+11x^{3}+3x^{5/2}+ \\
 +32x^{2}+3x^{3/2}+11x+\sqrt x +1 \\
 \end{array}} \right)^{2}.\notag
\end{align}
After simplifications, we have
\begin{align}
& h_{14} (x)=\left( {\sqrt x -1} \right)^{2}\times\notag\\
&\times \left( {\begin{array}{l}
 \left( {3x^{5}+2x^{5/2}+3} \right)\left( {\sqrt x -1}
\right)^{2}+48x^{9/2}+ \\
 +136x^{7/2}+136x^{5/2}+48x^{3/2}+4x^{6}+ \\
 +44x^{5}+56x^{4}+56x^{2}+44x+4 \\
 \end{array}} \right).\notag
\end{align}

Since $h_{14} (x)>0$, this gives that $k_{14} (x)>0$. Also we have
\[
\beta =\mathop {\sup }\limits_{x\in (0,\infty )}
g_{M_{2} I\mathunderscore JM_{1} } (x)=
\lim\limits_{x\to 1} g_{M_{2} I\mathunderscore JM_{1} } (x)=\textstyle{3
\over 7}.
\]

\item \textbf{For }$\bf{D_{M_{2} I}^{5} (P\vert \vert Q)\le 3D_{M_{1} I}^{2} (P\vert \vert Q)}$\textbf{: }Let us consider $g_{M_{2} I\mathunderscore M_{1} I} (x)={{f}''_{M_{2} I} (x)} \mathord{\left/ {\vphantom {{{f}''_{M_{2} I} (x)} {{f}''_{M_{1} I} (x)}}} \right. \kern-\nulldelimiterspace} {{f}''_{M_{1} I} (x)}$, then we have
\begin{align}
& g_{M_{2} I\mathunderscore M_{1} I} (x)=\notag\\
& =\frac{\sqrt {2x+2} \left( {2x-3\sqrt x +2} \right)-\left( {x^{(3/2)}+1}
\right)}{3\left[ {\sqrt {2x+2} \left( {x-\sqrt x +1} \right)-\left(
{x^{3/2}+1} \right)} \right]}\notag
\end{align}
and
\begin{align}
& {g}'_{M_{2} I\mathunderscore M_{1} I} (x)=\notag\\
& =-\frac{\left( {\sqrt x -1} \right)\times k_{15} (x)}{\left(
{\begin{array}{l}
 3\sqrt x \sqrt {2x+2} \left[ {-\left( {x^{3/2}+1} \right)+} \right. \\
 \left. {+\sqrt {2x+2} \left( {x-\sqrt x +1} \right)} \right]^{2} \\
 \end{array}} \right)},\notag
\end{align}
Where $k_{15} (x)=k_{8} (x)>0$. This give
\[
{g}'_{M_{2} I\mathunderscore M_{1} I} (x)\begin{cases}
 {>0,} & {x<1} \\
 {<0,} & {x>1} \\
\end{cases}.
\]
Also we have
\[
\beta =\mathop {\sup }\limits_{x\in (0,\infty )}
g_{M_{2} I\mathunderscore M_{1} I} (x)
=\lim\limits_{x\to 1} g_{M_{2} I\mathunderscore M_{1} I} (x)=3.
\]

\item \textbf{For }$\bf{D_{TJ}^{22} (P\vert \vert Q)\le \textstyle{8 \over {13}}D_{TM_{2} }^{25} (P\vert \vert Q)}$\textbf{: }Let us consider $g_{TJ\mathunderscore TM_{2} } (x)={{f}''_{TJ} (x)} \mathord{\left/ {\vphantom {{{f}''_{TJ} (x)} {{f}''_{TM_{2} } (x)}}} \right. \kern-\nulldelimiterspace} {{f}''_{TM_{2} } (x)}$, then we have
\[
g_{TJ\mathunderscore TM_{2} } (x)=\frac{3\left( {x-1} \right)^{2}\sqrt
{2x+2} }{2\left( {\begin{array}{l}
 2\sqrt x \left( {x^{3/2}+1} \right)+\sqrt {2x+2} \times \\
 \times \left( {3x^{2}-4x^{3/2}-4\sqrt x +3} \right) \\
 \end{array}} \right)}
\]
and
\begin{align}
& {g}'_{TJ\mathunderscore TM_{2} } (x)=\notag\\
& =-\frac{3\left( {x-1} \right)\times k_{16} (x)}{\left( {\begin{array}{l}
 \sqrt x \sqrt {2x+2} \left[ {2\sqrt x \left( {x^{3/2}+1} \right)} \right.+
\\
 \left. {+\sqrt {2x+2} \left( {3x^{2}-4x^{3/2}-4\sqrt x +3} \right)}
\right]^{2} \\
 \end{array}} \right)},\notag
\end{align}
where
\begin{align}
& k_{16} (x)=2\left( {x+1} \right)\sqrt {2x+2} \,\times\notag\\
& \times \left[ {\left( {\sqrt x -1} \right)^{4}+\sqrt x \left( {x+1} \right)}
\right] -\notag\\
& -\left( {\sqrt x +1} \right)\left( {\begin{array}{l}
 x^{2}\left( {\sqrt x -2} \right)^{2}+\left( {2\sqrt x -1} \right)^{2}+ \\
 +3\sqrt x \left( {x^{2}+1} \right) \\
 \end{array}} \right).\notag
\end{align}
This give
\[
{g}'_{TJ\mathunderscore TM_{2} } (x)\begin{cases}
 {>0} & {x<1} \\
 {<0} & {x>1} \\
\end{cases},
\]
provided $k_{16} (x)>0$. In order to prove $k_{16} (x)>0$, let us consider
\begin{align}
& h_{16} (x)=\left( {2\left( {x+1} \right)\sqrt {2x+2} } \right)^{2}\times\notag\\
& \times \left[ {\left( {\sqrt x -1} \right)^{4}+\sqrt x \left( {x+1} \right)}
\right]-\left( {\sqrt x +1} \right)^{2}\times\notag\\
& \times \left( {\begin{array}{l}
 x^{2}\left( {\sqrt x -2} \right)^{2}+\left( {2\sqrt x -1} \right)^{2}+ \\
 +3\sqrt x \left( {x^{2}+1} \right) \\
 \end{array}} \right)^{2}.\notag
\end{align}
After simplifications, we have
\begin{align}
& h_{16} (x)=\left( {\sqrt x -1} \right)^{4}\times\notag\\
& \times \left( {\begin{array}{l}
 x\left( {42x^{2}+65x+42} \right)\left( {\sqrt x -1} \right)^{2}+ \\
 +5x^{4}\left( {\sqrt x -2} \right)^{2}+5\left( {2\sqrt x -1} \right)^{2}+
\\
 +\left( {x+1} \right)\left( {2x^{4}+45x^{2}+2} \right) \\
 \end{array}} \right).\notag
\end{align}
Since $h_{16} (x)>0$, this gives that $k_{16} (x)>0$. Also we have
\[
\beta =\mathop {\sup }\limits_{x\in (0,\infty )}
g_{TJ\mathunderscore TM_{2} } (x)
=\lim\limits_{x\to 1} g_{TJ\mathunderscore TM_{2} } (x)=\textstyle{8
\over {13}}.
\]

\item \textbf{For }$\bf{D_{TM_{1} }^{26} (P\vert \vert Q)\le \textstyle{{15} \over {13}}D_{TM_{2} }^{25} (P\vert \vert Q)}$\textbf{: }Let us consider $g_{TM_{1} \mathunderscore TM_{2} } (x)={{f}''_{TM_{1} } (x)} \mathord{\left/ {\vphantom {{{f}''_{TM_{1} } (x)} {{f}''_{TM_{2} } (x)}}} \right. \kern-\nulldelimiterspace} {{f}''_{TM_{2} } (x)}$, then we have
\begin{align}
& g_{TM_{1} \mathunderscore TM_{2} } (x)=\notag\\
& =\frac{\left( {\begin{array}{l}
 3\sqrt {2x+2} \left( {x^{2}-2x^{3/2}-2\sqrt x +1} \right)+ \\
 +2\sqrt x \left( {x^{3/2}+1} \right) \\
 \end{array}} \right)}{\left( {\begin{array}{l}
 \sqrt {2x+2} \left( {3x^{2}-4x^{3/2}-4\sqrt x +3} \right)+ \\
 +2\sqrt x \left( {x^{3/2}+1} \right) \\
 \end{array}} \right)}\notag
\end{align}
and
\begin{align}
& {g}'_{TM_{1} \mathunderscore TM_{2} } (x)=\notag\\
& =-\frac{6\left( {\sqrt x -1} \right)k_{17} (x)}{\left( {\begin{array}{l}
 \sqrt x \sqrt {2x+2} \left[ {2\sqrt x \left( {x^{3/2}+1} \right)+} \right.
\\
 \left. {+\sqrt {2x+2} \left( {3x^{2}-4x^{3/2}-4\sqrt x +3} \right)}
\right]^{2} \\
 \end{array}} \right)},\notag
\end{align}
where
\begin{align}
k_{17} (x)=& -\left( {\begin{array}{l}
 \sqrt {2x+2} \left( {\sqrt x +1} \right)\times \\
 \times \left( {x+1} \right)\left( {x^{2}+4x+1} \right) \\
 \end{array}} \right)+\notag\\
& +2\left( {\begin{array}{l}
 x^{4}+x^{7/2}+x^{3}+8x^{5/2}+ \\
 +2x^{2}+8x^{3/2}+x+\sqrt x +1 \\
 \end{array}} \right).\notag
\end{align}

This give
\[
{g}'_{TM_{1} \mathunderscore TM_{2} } (x)\begin{cases}
 {>0,} & {x<1} \\
 {<0,} & {x>1} \\
\end{cases},
\]
provided $k_{17} (x)>0$. In order to show $k_{17} (x)>0$, let us consider
\begin{align}
& h_{17} (x)=-\left( {\begin{array}{l}
 \sqrt {2x+2} \left( {\sqrt x +1} \right)\times \\
 \times \left( {x+1} \right)\left( {x^{2}+4x+1} \right) \\
 \end{array}} \right)^{2}\notag\\
& +\left[ {2\left( {\begin{array}{l}
 x^{4}+x^{7/2}+x^{3}+8x^{5/2}+ \\
 \mathunderscore 2x^{2}+8x^{3/2}+x+\sqrt x +1 \\
 \end{array}} \right)} \right]^{2}.\notag
\end{align}
After simplifications, we have
\begin{align}
& h_{17} (x)=2\left( {\sqrt x -1} \right)^{4}\times\notag\\
& \times \left( {\begin{array}{l}
 1+6\sqrt x +30x^{3/2}+72x^{5/2}+12x^{5}+ \\
 +57x^{2}+94x^{3}+57x^{4}+x^{6}+ \\
 +12x+6x^{11/2}+72x^{7/2}+30x^{9/2} \\
 \end{array}} \right).\notag
\end{align}

Since $h_{17} (x)>0$, proving that $k_{17} (x)>0$. Also we have
\[
\beta =\mathop {\sup }\limits_{x\in (0,\infty )}
g_{TM_{1} \mathunderscore TM_{2} } (x)=
\lim\limits_{x\to 1} g_{TM_{1} \mathunderscore TM_{2} }
(x)=\textstyle{{15} \over {13}}.
\]

\item \textbf{For }$\bf{D_{TM_{2} }^{25} (P\vert \vert Q)\le \textstyle{{13} \over {12}}D_{Th}^{24} (P\vert \vert Q)}$\textbf{: }Let us consider $g_{TM_{2} \mathunderscore Th} (x)={{f}''_{TM_{2} } (x)} \mathord{\left/ {\vphantom {{{f}''_{TM_{2} } (x)} {{f}''_{Th} (x)}}} \right. \kern-\nulldelimiterspace} {{f}''_{Th} (x)}$, then we have
\begin{align}
& g_{TM_{2} \mathunderscore Th} (x)=\notag\\
& = \frac{\left( {\begin{array}{l}
 \sqrt {2x+2} \left( {3x^{2}-4x^{3/2}-4\sqrt x +3} \right)+ \\
 +2\sqrt x \left( {x^{3/2}+1} \right) \\
 \end{array}} \right)}{3 \sqrt {2x+2} \left( {x+\sqrt x +1} \right)\left(
{\sqrt x -1} \right)^{2}}\notag
\end{align}
and
\begin{align}
& {g}'_{TM_{2} \mathunderscore Th} (x)=\notag
-\frac{k_{18} (x)}{\left( {\begin{array}{l}
 6\left( {\sqrt x -1} \right)^{3}\sqrt x \left( {x+1} \right)\times \\
 \times \sqrt {2x+2} \left( {x+\sqrt x +1} \right)^{2} \\
 \end{array}} \right)},\notag
\end{align}
where $k_{18} (x)=k_{17} (x)>0$. This gives
\[
{g}'_{TM_{2} \mathunderscore Th} (x)\begin{cases}
 {>0,} & {x<1} \\
 {<0,} & {x>1} \\
\end{cases}.
\]
Also we have
\[
\beta =\mathop {\sup }\limits_{x\in (0,\infty )}
g_{TM_{2} \mathunderscore Th} (x)=
\lim\limits_{x\to 1} g_{TM_{2} \mathunderscore Th} (x)=\textstyle{{13}
\over {12}}.
\]

\item \textbf{For }$\bf{D_{Th}^{24} (P\vert \vert Q)\le \textstyle{4 \over 3}D_{TM_{3} }^{23} (P\vert \vert Q)}$\textbf{: }Let us consider $g_{Th\mathunderscore TM_{3} } (x)={{f}''_{Th} (x)} \mathord{\left/ {\vphantom {{{f}''_{Th} (x)} {{f}''_{TM_{3} } (x)}}} \right. \kern-\nulldelimiterspace} {{f}''_{TM_{3} } (x)}$, then we have
\begin{align}
& g_{Th\mathunderscore TM_{3} } (x)=\notag\\
& =\frac{\left( {x+\sqrt x +1} \right)\left( {\sqrt x -1} \right)^{2}\sqrt
{2x+2} }{\left( {x^{2}+1} \right)\sqrt {2x+2} -2\sqrt x \left( {x^{3/2}+1}
\right)}\notag
\end{align}
and
\begin{align}
& {g}'_{Th\mathunderscore TM_{3} } (x)=\notag\\
& =-\frac{\left( {\sqrt x -1} \right)\times k_{19} (x)}{\left(
{\begin{array}{l}
 \sqrt x \sqrt {2x+2} \left[ {\left( {x^{2}+1} \right)\sqrt {2x+2} }
\right.- \\
 \left. {-2\sqrt x \left( {x^{3/2}+1} \right)} \right]^{2} \\
 \end{array}} \right)},\notag
\end{align}
where $k_{19} (x)=k_{17} (x)>0$. This gives
\[
{g}'_{Th\mathunderscore TM_{3} } (x)\begin{cases}
 {>0,} & {x<1} \\
 {<0,} & {x>1} \\
\end{cases}.
\]
Also we have
\[
\beta =\mathop {\sup }\limits_{x\in (0,\infty )}
g_{Th\mathunderscore TM_{3} } (x)=
\lim\limits_{x\to 1} g_{Th\mathunderscore TM_{3} } (x)=\textstyle{4
\over 3}.
\]

\item \textbf{For }$\bf{D_{Th}^{24} (P\vert \vert Q)\le \textstyle{{12} \over 5}D_{JM_{2} }^{18} (P\vert \vert Q)}$\textbf{: }Let us consider $g_{Th\mathunderscore JM_{2} } (x)={{f}''_{Th} (x)} \mathord{\left/ {\vphantom {{{f}''_{Th} (x)} {{f}''_{JM_{2} } (x)}}} \right. \kern-\nulldelimiterspace} {{f}''_{JM_{2} } (x)}$, then we have
\begin{align}
& g_{Th\mathunderscore JM_{2} } (x)=\notag\\
& =\frac{6\left( {x+\sqrt x +1} \right)\left( {\sqrt x -1} \right)^{2}\sqrt x
\sqrt {2x+2} }{\left( {\begin{array}{l}
 \sqrt {2x+2} \sqrt x \left( {x+1} \right)\times \\
 \times \left( {3x-8\sqrt x +3} \right)+4x\left( {x^{3/2}+1} \right) \\
 \end{array}} \right)}\notag
\end{align}
and
\begin{align}
& {g}'_{Th\mathunderscore JM_{2} } (x)=\notag\\
& =-\frac{3\left( {\sqrt x -1} \right)\sqrt x \sqrt {2x+2} \times k_{20}
(x)}{\left( {\begin{array}{l}
 \left( {x+1} \right)\left[ {\sqrt {2x+2} \sqrt x \left( {x+1} \right)\times
} \right. \\
 \left. {\times \left( {3x-8\sqrt x +3} \right)+4x\left( {x^{3/2}+1}
\right)} \right]^{`2} \\
 \end{array}} \right)},\notag
\end{align}
where
\begin{align}
& k_{20} (x)=\sqrt {2x+2} \left( {\sqrt x +1} \right)\left( {x+1}
\right)\times\notag\\
& \times \left( {\begin{array}{l}
 x\left( {\sqrt x -2} \right)^{2}+\left( {2\sqrt x -1} \right)^{2}+ \\
 +4\left( {x+1} \right)\left( {\sqrt x -1} \right)^{2}+10x \\
 \end{array}} \right)-\notag\\
& -4\left( {\begin{array}{l}
 x^{4}+x^{7/2}+x^{3}+8x^{5/2}+ \\
 +2x^{2}+8x^{3/2}+x+\sqrt x +1 \\
 \end{array}} \right).\notag
\end{align}
This gives
\[
{g}'_{Th\mathunderscore JM_{2} } (x)\begin{cases}
 {>0} & {x<1} \\
 {<0} & {x>1} \\
\end{cases},
\]
provided $k_{20} (x)>0$. In order to show $k_{20} (x)>0$, let us consider
\begin{align}
& h_{20} (x)=\left[ {\sqrt {2x+2} \left( {\sqrt x +1} \right)\left( {x+1}
\right)} \right]^{2}\times\notag\\
& \times \left( {\begin{array}{l}
 x\left( {\sqrt x -2} \right)^{2}+\left( {2\sqrt x -1} \right)^{2}+ \\
 +4\left( {x+1} \right)\left( {\sqrt x -1} \right)^{2}+10x \\
 \end{array}} \right)^{2}-\notag\\
& -\left[ {4\left( {\begin{array}{l}
 x^{4}+x^{7/2}+x^{3}+8x^{5/2}+ \\
 +2x^{2}+8x^{3/2}+x+\sqrt x +1 \\
 \end{array}} \right)} \right]^{2}.\notag
\end{align}

After simplifications, we have
\begin{align}
& h_{20} (x)=2\left( {\sqrt x -1} \right)^{4}\times\notag\\
& \times \left( {\begin{array}{l}
 \left( {\sqrt x -1} \right)^{2}\left( {9+14x^{5/2}+9x^{5}} \right)+57x+ \\
 +118x^{7/2}+8x^{6}+8+118x^{5/2}+57x^{5}+ \\
 +39x^{2}+39x^{4}+30x^{3/2}+30x^{9/2} \\
 \end{array}} \right).\notag
\end{align}
Since $h_{20} (x)>0$, proving that $k_{20} (x)>0$. Also we have
\[
\beta =\mathop {\sup }\limits_{x\in (0,\infty )}
g_{Th\mathunderscore JM_{2} } (x)=
\lim\limits_{x\to 1} g_{Th\mathunderscore JM_{2} } (x)=\textstyle{{12}
\over 5}.
\]

\item \textbf{For }$\bf{D_{JM_{2} }^{18} (P\vert \vert Q)\le \textstyle{5 \over 4}D_{Jh}^{17} (P\vert \vert Q)}$\textbf{: }Let us consider $g_{JM_{2} \mathunderscore Jh} (x)={{f}''_{JM_{2} } (x)} \mathord{\left/ {\vphantom {{{f}''_{JM_{2} } (x)} {{f}''_{Jh} (x)}}} \right. \kern-\nulldelimiterspace} {{f}''_{Jh} (x)}$, then we have
\begin{align}
& g_{JM_{2} \mathunderscore Jh} (x)=\notag\\
& =\frac{\left( {\begin{array}{l}
 \sqrt {2x+2} \sqrt x \left( {x+1} \right)\times \\
 \times \left( {3x-8\sqrt x +3} \right)+4x\left( {x^{3/2}+1} \right) \\
 \end{array}} \right)}{3\sqrt {2x+2} \sqrt x \left( {x+1} \right)\left(
{\sqrt x -1} \right)^{2}}\notag
\end{align}
and
\[
{g}'_{JM_{2} \mathunderscore Jh} (x)=-\frac{k_{21} (x)}{3\sqrt {2x+2} \sqrt
x \left( {x+1} \right)\left( {\sqrt x -1} \right)^{3}},
\]
where
\begin{align}
& k_{21} (x)=2\left[ {x^{3}+\sqrt x \left( {x+1} \right)\left( {\sqrt x -1}
\right)^{2}} \right]+\notag\\
& +2\left( {6x^{3/2}+1} \right)-\sqrt {2x+2} \left( {\sqrt x +1} \right)\left(
{x+1} \right)^{2}.\notag
\end{align}
This gives
\[
{g}'_{JM_{2} \mathunderscore Jh} (x)\begin{cases}
 {>0,} & {x<1} \\
 {<0,} & {x>1} \\
\end{cases}.
\]
provided $k_{21} (x)>0$. In order to prove $k_{21} (x)>0$, let us consider
\begin{align}
& h_{21} (x)=4\left( {\begin{array}{l}
 \sqrt x \left( {x+1} \right)\left( {\sqrt x -1} \right)^{2}+ \\
 +x^{3}+2\left( {6x^{3/2}+1} \right) \\
 \end{array}} \right)^{2}-\notag\\
& -\left[ {\sqrt {2x+2} \left( {\sqrt x +1} \right)\left( {x+1} \right)^{2}}
\right]^{2}.\notag
\end{align}
After simplifications, we have
\begin{align}
& h_{21} (x)=2\left( {\sqrt x -1} \right)^{4}\times \notag\\
& \left( {\begin{array}{l}
 x^{4}+6x^{7/2}+6x^{3}+6x^{5/2}+ \\
 +28x^{2}+6x^{3/2}+6x+6\sqrt x +1 \\
 \end{array}} \right).\notag
\end{align}
Since $h_{21} (x)>0$, this gives that $k_{21} (x)>0$ Also we have
\[
\beta =\mathop {\sup }\limits_{x\in (0,\infty )}
g_{JM_{2} \mathunderscore Jh} (x)=
\lim \limits_{x\to 1} g_{JM_{2} \mathunderscore Jh} (x)=\textstyle{5
\over 4}.
\]

\item \textbf{For }$\bf{D_{K_{0} \Delta }^{36} (P\vert \vert Q)\le \textstyle{3 \over 2}D_{K_{0} I}^{35} (P\vert \vert Q)}$\textbf{:} Let us consider $g_{K_{0} \Delta \mathunderscore K_{0} I} (x)={{f}''_{K_{0} \Delta } (x)} \mathord{\left/ {\vphantom {{{f}''_{K_{0} \Delta } (x)} {{f}''_{K_{0} I} (x)}}} \right. \kern-\nulldelimiterspace} {{f}''_{K_{0} I} (x)}$, then we have
\begin{align}
& g_{K_{0} \Delta \mathunderscore K_{0} I} (x)=\notag\\
& =\frac{\left( {\begin{array}{l}
 3x^{4}+6x^{7/2}+20x^{3}+34x^{5/2}+ \\
 +\,66x^{2}+34x^{3/2}+20x+6\sqrt x +3 \\
 \end{array}} \right)}{(x+1)^{2}\left( {3x^{2}+6x^{3/2}+14x+6\sqrt x +3}
\right)}\notag
\end{align}
and
\begin{align}
& {g}'_{K_{0} \Delta \mathunderscore K_{0} I} (x)=\notag\\
& =-\frac{8\left( {x-1} \right)\sqrt x \left( {3x+8\sqrt x +3} \right)}{\left(
{x+1} \right)^{3}}\times\notag\\
& \times \frac{\left( {3x^{2}+4x^{3/2}+10x+4\sqrt x +3} \right)}{\left(
{3x^{2}+6x^{3/2}+14x+6\sqrt x +3} \right)^{2}}.\notag
\end{align}

This gives
\[
{g}'_{K_{0} \Delta \mathunderscore K_{0} I} (x)\begin{cases}
 {>0,} & {x<1} \\
 {<0,} & {x>1} \\
\end{cases}.
\]
Also we have
\[
\beta =\mathop {\sup }\limits_{x\in
(0,\infty )} g_{K_{0} \Delta \mathunderscore K_{0} I} (x)=
\lim\limits_{x\to 1} g_{K_{0} \Delta \mathunderscore K_{0} I}
(x)=\textstyle{3 \over 2}.
\]

\item \textbf{For }$\bf{D_{TM_{3} }^{23} (P\vert \vert Q)\le \textstyle{9 \over {16}}D_{K_{0} I}^{35} (P\vert \vert Q)}$\textbf{: }Let us consider $g_{TM_{3} \mathunderscore K_{0} I} (x)={{f}''_{TM_{3} } (x)} \mathord{\left/ {\vphantom {{{f}''_{TM_{3} } (x)} {{f}''_{K_{0} I} (x)}}} \right. \kern-\nulldelimiterspace} {{f}''_{K_{0} I} (x)}$, then we have
\begin{align}
& g_{TM_{3} \mathunderscore K_{0} I} (x)=\notag\\
& =\frac{\left[ {\sqrt {2x+2} \left( {x^{2}+1} \right)-2\sqrt x \left(
{x^{3/2}+1} \right)} \right]}{\sqrt {2x+2} \left( {\sqrt x -1}
\right)^{2}}\times\notag\\
& \times \frac{8\sqrt x }{\left( {3x^{2}+6x^{3/2}+14x+6\sqrt x +3} \right)}\notag
\end{align}
and
\begin{align}
& {g}'_{TM_{3} \mathunderscore K_{0} I} (x)=\notag\\
& =-\frac{4\times k_{22} (x)}{\left( {\begin{array}{l}
 \sqrt x \sqrt {2x+2} \left( {x+1} \right)\left( {\sqrt x -1}
\right)^{3}\times \\
 \times \left( {3x^{2}+6x^{3/2}+14x+6\sqrt x +3} \right)^{2} \\
 \end{array}} \right)},\notag
\end{align}

where
\begin{align}
& k_{22} (x)=\sqrt {2x+2} \, \times\notag\\
& \times \left( {\begin{array}{l}
 3x^{11/2}+28x^{2}+3\sqrt x +33x^{3/2}+ \\
 +60x^{3}+33x^{4}+3x^{5}+60x^{5/2}+ \\
 +28x^{7/2}+x^{9/2}+x+3 \\
 \end{array}} \right)-\notag\\
& -2\sqrt x \left( {\begin{array}{l}
 6\sqrt x +40x^{3/2}+40x^{7/2}+ \\
 +9x+84x^{5/2}+6x^{5}+9x^{4}+ \\
 +25x^{3}+6x^{9/2}+25x^{2}+6 \\
 \end{array}} \right).\notag
\end{align}

This gives
\[
{g}'_{TM_{3} \mathunderscore K_{0} I} (x)\begin{cases}
 {>0,} & {x<1} \\
 {<0,} & {x>1} \\
\end{cases},
\]
provided $k_{22} (x)>0$. In order to prove $k_{22} (x)>0$, let us consider
\begin{align}
& h_{22} (x)=\left( {\sqrt {2x+2} } \right)^{2}\times\notag\\
& \times \left( {\begin{array}{l}
 3x^{11/2}+28x^{2}+3\sqrt x +33x^{3/2}+ \\
 +60x^{3}+33x^{4}+3x^{5}+60x^{5/2}+ \\
 +28x^{7/2}+x^{9/2}+x+3 \\
 \end{array}} \right)^{2} \notag\\
& -\left[ {2\sqrt x \left( {\begin{array}{l}
 6\sqrt x +40x^{3/2}+40x^{7/2}+ \\
 +9x+84x^{5/2}+6x^{5}+9x^{4}+ \\
 +25x^{3}+6x^{9/2}+25x^{2}+6 \\
 \end{array}} \right)} \right]^{2}.\notag
\end{align}

After simplification, we have
\begin{align}
& h_{22} (x)=2\left( {\sqrt x -1} \right)^{4}\times\notag\\
& \times \left( {\begin{array}{l}
 9+54\sqrt x +1326x^{4}+54x^{19/2}+ \\
 +246x^{17/2}+1324x^{9/2}+1324x^{11/2}+ \\
 +952x^{15/2}+1808x^{13/2}+1582x^{3}+ \\
 +952x^{5/2}+114x+601x^{2}+352x^{5}+ \\
 +9x^{10}+114x^{9}+601x^{8}+1582x^{7}+ \\
 +1326x^{6}+1808x^{7/2}+246x^{3/2} \\
 \end{array}} \right).\notag
\end{align}

Since $h_{22} (x)>0$, proving that $k_{22} (x)>0$. Also we have
\[
\beta =\mathop {\sup }\limits_{x\in (0,\infty )}
g_{TM_{3} \mathunderscore K_{0} I} (x)=
\lim\limits_{x\to 1} g_{TM_{3} \mathunderscore K_{0} I} (x)=\textstyle{9
\over {16}}.
\]

\item \textbf{For }$\bf{D_{M_{1} I}^{2} (P\vert \vert Q)\le \textstyle{1 \over {16}}D_{K_{0} I}^{35} (P\vert \vert Q)}$\textbf{: }Let us consider $g_{M_{1} I\mathunderscore K_{0} I} (x)={{f}''_{M_{1} I} (x)} \mathord{\left/ {\vphantom {{{f}''_{M_{1} I} (x)} {{f}''_{K_{0} I} (x)}}} \right. \kern-\nulldelimiterspace} {{f}''_{K_{0} I} (x)}$, then we have
\begin{align}
& g_{M_{1} I\mathunderscore K_{0} I} (x)=\notag\\
& =\frac{16x\left[ {\sqrt {2x+2} \left( {x-\sqrt x +1} \right)-\left(
{x^{3/2}+1} \right)} \right]}{\left( {\begin{array}{l}
 \sqrt {2x+2} \left( {\sqrt x -1} \right)^{2}\times \\
 \times \left( {3x^{2}+6x^{3/2}+14x+6\sqrt x +3} \right) \\
 \end{array}} \right)}\notag
\end{align}
and
\begin{align}
& {g}'_{M_{1} I\mathunderscore K_{0} I} (x)=\notag\\
& =-\frac{8\times k_{23} (x)}{\left( {\begin{array}{l}
 \left( {\sqrt x -1} \right)^{3}\sqrt {2x+2} \left( {x+1} \right)\times \\
 \times \left( {3x^{2}+6x^{3/2}+14x+6\sqrt x +3} \right)^{2} \\
 \end{array}} \right)},\notag
\end{align}
where
\begin{align}
& k_{23} (x)=\sqrt {2x+2} \left( {x+1} \right)\left( {\sqrt x +1}
\right)\times \notag\\
& \times \left( {\begin{array}{l}
 6\left( {x^{2}+1} \right)\left( {\sqrt x -1} \right)^{2}+2x^{3/2} \\
 +3x^{2}\left( {\sqrt x +4} \right)+3\sqrt x \left( {4\sqrt x +1} \right) \\
 \end{array}} \right)-\notag\\
& -\left( {\begin{array}{l}
 6x^{9/2}+40x^{7/2}+84x^{5/2}+40x^{3/2}+6\sqrt x \\
 +6x^{5}+9x^{4}+25x^{3}+25x^{2}+9x+6 \\
 \end{array}} \right).\notag
\end{align}
This gives
\[
{g}'_{M_{1} I\mathunderscore K_{0} I} (x)\begin{cases}
 {>0,} & {x<1} \\
 {<0,} & {x>1} \\
\end{cases},
\]
provided $k_{23} (x)>0$. In order to prove $k_{23} (x)>0$, let us consider
\begin{align}
& h_{23} (x)=\left[ {\sqrt {2x+2} \left( {x+1} \right)\left( {\sqrt x +1}
\right)} \right]^{2}\times\notag\\
& \times \left( {\begin{array}{l}
 6\left( {x^{2}+1} \right)\left( {\sqrt x -1} \right)^{2}+2x^{3/2}+ \\
 +3x^{2}\left( {\sqrt x +4} \right)+3\sqrt x \left( {4\sqrt x +1} \right) \\
 \end{array}} \right)^{2}-\notag\\
& -\left( {\begin{array}{l}
 6x^{9/2}+40x^{7/2}+84x^{5/2}+40x^{3/2}+6\sqrt x \\
 +6x^{5}+9x^{4}+25x^{3}+25x^{2}+9x+6 \\
 \end{array}} \right)^{2}.\notag
\end{align}

After simplification, we get
\begin{align}
& h_{23} (x)=\left( {\sqrt x -1} \right)^{6}\times\notag\\
& \times \left( {\begin{array}{l}
 36+72\sqrt x +198x+1314x^{5/2}+ \\
 +396x^{3/2}+1314x^{9/2}+396x^{11/2}+ \\
 +1717x^{3}+765x^{2}+2012x^{7/2}+ \\
 +72x^{13/2}+765x^{5}+1717x^{4}+ \\
 +36x^{7}+198x^{6} \\
 \end{array}} \right).\notag
\end{align}

Since $h_{23} (x)>0$, this gives that $k_{23} (x)>0$ Also we have
\[
\beta =\mathop {\sup }\limits_{x\in (0,\infty )}
g_{M_{1} I\mathunderscore K_{0} I} (x)=
\lim\limits_{x\to 1} g_{M_{1} I\mathunderscore K_{0} I} (x)=\textstyle{1
\over {16}}.
\]

\item \textbf{For }$\bf{D_{Jh}^{17} (P\vert \vert Q)\le \textstyle{1 \over 4}D_{K_{0} I}^{35} (P\vert \vert Q)}$\textbf{: }Let us consider \newline $g_{Jh\mathunderscore K_{0} I} (x)={{f}''_{Jh} (x)} \mathord{\left/ {\vphantom {{{f}''_{Jh} (x)} {{f}''_{K_{0} I} (x)}}} \right. \kern-\nulldelimiterspace} {{f}''_{K_{0} I} (x)}$, then we have
\[
g_{Jh\mathunderscore K_{0} I} (x)=\frac{4\sqrt x \left( {x+1}
\right)}{3x^{2}+6x^{3/2}+14x+6\sqrt x +3}
\]
\[
{g}'_{Jh\mathunderscore K_{0} I} (x)=-\frac{2\left( {x-1} \right)\left[
{2(x^{2}+1)+(x-1)^{2}} \right]}{\left( {3x^{2}+6x^{3/2}+14x+6\sqrt x +3}
\right)^{2}}.
\]

This gives
\[
{g}'_{Jh\mathunderscore K_{0} I} (x)\begin{cases}
 {>0,} & {x<1} \\
 {<0,} & {x>1} \\
\end{cases}.
\]
Also we have
\[
\beta =\mathop {\sup }\limits_{x\in (0,\infty )}
g_{Jh\mathunderscore K_{0} I} (x)=
\lim\limits_{x\to 1} g_{Jh\mathunderscore K_{0} I} (x)=\textstyle{1
\over 4}.
\]

\item \textbf{For }$\bf{D_{K_{0} I}^{35} (P\vert \vert Q)\le \textstyle{{16} \over {15}}D_{K_{0} M_{1} }^{34} (P\vert \vert Q)}$\textbf{: }Let us consider $g_{K_{0} I\mathunderscore K_{0} M_{1} } (x)={{f}''_{K_{0} I} (x)} \mathord{\left/ {\vphantom {{{f}''_{K_{0} I} (x)} {{f}''_{K_{0} M_{1} } (x)}}} \right. \kern-\nulldelimiterspace} {{f}''_{K_{0} M_{1} } (x)}$, then we have
\begin{align}
& g_{K_{0} I\mathunderscore K_{0} M_{1} } (x)=\notag\\
& =\frac{\left( {\begin{array}{l}
 \sqrt {2x+2} \left( {\sqrt x -1} \right)^{2}\times \\
 \times \left( {3x^{2}+6x^{3/2}+14x+6\sqrt x +3} \right) \\
 \end{array}} \right)}{\left( {\begin{array}{l}
 \sqrt {2x+2} \left( {x+1} \right)\times \\
 \times \left( {3x^{2}-14x+3} \right)+16x\left( {x^{3/2}+1} \right) \\
 \end{array}} \right)}\notag
\end{align}
and
\begin{align}
& {g}'_{K_{0} I\mathunderscore K_{0} M_{1} } (x)=\notag\\
& =-\frac{16\left( {\sqrt x -1} \right)\times k_{24} (x)}{\left(
{\begin{array}{l}
 \sqrt x \sqrt {2x+2} \left[ {16x\left( {x^{3/2}+1} \right)+} \right. \\
 \left. {+\sqrt {2x+2} \left( {x+1} \right)\left( {3x^{2}-14x+3} \right)}
\right]^{2} \\
 \end{array}} \right)}, \notag
\end{align}
where $k_{24} (x)=k_{23} (x)>0$. This gives
\[
{g}'_{K_{0} I\mathunderscore K_{0} M_{1} } (x)\begin{cases}
 {>0} & {x<1} \\
 {<0} & {x>1} \\
\end{cases}.
\]
Also we have
\[
\beta=\mathop {\sup }\limits_{x\in (0,\infty
)} g_{K_{0} I\mathunderscore K_{0} M_{1} } (x)=
\]
\[
\quad = \lim\limits_{x\to 1} g_{K_{0} I\mathunderscore K_{0} M_{1} }
(x)=\textstyle{{16} \over {15}}.
\]

\item \textbf{For}$\bf{D_{K_{0} M_{1} }^{34} (P\vert \vert Q)\le \textstyle{{15} \over {13}}D_{K_{0} M_{2} }^{33} (P\vert \vert Q)}$\textbf{: }Let us consider $g_{K_{0} M_{1} \mathunderscore K_{0} M_{2} } (x)={{f}''_{K_{0} M_{1} } (x)} \mathord{\left/ {\vphantom {{{f}''_{K_{0} M_{1} } (x)} {{f}''_{K_{0} M_{2} } (x)}}} \right. \kern-\nulldelimiterspace} {{f}''_{K_{0} M_{2} } (x)}$, then we have
\begin{align}
& g_{K_{0} M_{1} \mathunderscore K_{0} M_{2} } (x)=\notag\\
& =\frac{3\left( {\begin{array}{l}
 \sqrt {2x+2} \left( {x+1} \right)\times \\
 \times \left( {3x^{2}-14x+3} \right)+16x\left( {x^{3/2}+1} \right) \\
 \end{array}} \right)}{\left( {\begin{array}{l}
 \sqrt {2x+2} \left( {x+1} \right)\times \\
 \times \left( {9x^{2}-26x+9} \right)+16x\left( {x^{3/2}+1} \right) \\
 \end{array}} \right)}\notag
\end{align}
and
\begin{align}
& {g}'_{K_{0} M_{1} \mathunderscore K_{0} M_{2} } (x)=\notag\\
& =-\frac{288(x-1)(x+1)\times k_{25} (x)}{\left( {\begin{array}{l}
 \sqrt {2x+2} \left( {x+1} \right)\left( {9x^{2}-26x+9} \right)+ \\
 +16x\left( {x^{3/2}+1} \right) \\
 \end{array}} \right)^{2}},\notag
\end{align}
where
\begin{align}
& k_{25} (x)=\left( {2x^{7/2}+x^{5/2}+5x^{2}+5x^{3/2}+x+2} \right)-\notag\\
& -\sqrt {2x+2} \left( {x+1} \right)^{3}.\notag
\end{align}
This gives
\[
{g}'_{K_{0} M_{1} \mathunderscore K_{0} M_{2} } (x)\begin{cases}
 {>0} & {x<1} \\
 {<0} & {x>1} \\
\end{cases},
\]
provided $k_{25} (x)>0$. In order to prove $k_{25} (x)>0$, let us consider a
function
\begin{align}
& h_{25} (x)=\left( {\begin{array}{l}
 2x^{7/2}+x^{5/2}+5x^{2}+ \\
 +5x^{3/2}+x+2 \\
 \end{array}} \right)^{2}- \notag\\
& -\left[ {\sqrt {2x+2} \left( {x+1} \right)^{3}} \right]^{2}.\notag
\end{align}

After simplifications, we get
\begin{align}
& h_{25} (x)=\left( {\sqrt x -1} \right)^{4}\times \notag\\
& \times \left( {\begin{array}{l}
 2x^{5}+8x^{9/2}+10x^{4}+20x^{7/2}+ \\
 +29x^{3}+42x^{5/2}+29x^{2}+ \\
 +20x^{3/2}+8\sqrt x +10x+2 \\
 \end{array}} \right).\notag
\end{align}

Since $h_{25} (x)>0$, this gives that $k_{25} (x)>0$. Also we have
\[
\beta =\mathop {\sup }\limits_{x\in
(0,\infty )} g_{K_{0} M_{1} \mathunderscore K_{0} M_{2} } (x)=
\]
\[
\quad = \lim\limits_{x\to 1} g_{K_{0} M_{1} \mathunderscore K_{0} M_{2} }
(x)=\textstyle{{15} \over {13}}.
\]

\item \textbf{For }$\bf{D_{K_{0} M_{2} }^{33} (P\vert \vert Q)\le \textstyle{{13} \over {12}}D_{K_{0} h\,}^{32} (P\vert \vert Q)}$\textbf{: }Let us consider $g_{K_{0} M_{2} \mathunderscore K_{0} h} (x)={{f}''_{K_{0} M_{2} } (x)} \mathord{\left/ {\vphantom {{{f}''_{K_{0} M_{2} } (x)} {{f}''_{K_{0} h} (x)}}} \right. \kern-\nulldelimiterspace} {{f}''_{K_{0} h} (x)}$, then we have
\begin{align}
& g_{K_{0} M_{2} \mathunderscore K_{0} h} (x)=\notag\\
& =\frac{\left( {\begin{array}{l}
 \sqrt {2x+2} \left( {x+1} \right)\times \\
 \times \left( {9x^{2}-26x+9} \right)+16x\left( {x^{3/2}+1} \right) \\
 \end{array}} \right)}{9\sqrt {2x+2} \left( {x+1} \right)\left( {x-1}
\right)^{2}}\notag
\end{align}
and
\[
{g}'_{K_{0} M_{2} \mathunderscore K_{0} h} (x)=-\frac{8\times k_{26}
(x)}{9\left( {x-1} \right)^{3}\left( {x+1} \right)^{2}\sqrt {2x+2} },
\]
where $k_{26} (x)=k_{25} (x)>0$ This gives
\[
{g}'_{K_{0} M_{2} \mathunderscore K_{0} h} (x)\begin{cases}
 {>0,} & {x<1} \\
 {<0,} & {x>1} \\
\end{cases}.
\]
Also, we have
\[
\beta =\mathop {\sup }\limits_{x\in (0,\infty
)} g_{K_{0} M_{2} \mathunderscore K_{0} h} (x)=
\]
\[
\quad = \lim\limits_{x\to 1} g_{K_{0} M_{2} \mathunderscore K_{0} h}
(x)=\textstyle{{13} \over {12}}.
\]

\item \textbf{For }$\bf{D_{K_{0} h\,}^{32} (P\vert \vert Q)\le \textstyle{4 \over 3}D_{K_{0} M_{3} }^{31} (P\vert \vert Q)}$\textbf{: }Let us consider $g_{K_{0} h\mathunderscore K_{0} M_{3} } (x)={{f}''_{K_{0} h} (x)} \mathord{\left/ {\vphantom {{{f}''_{K_{0} h} (x)} {{f}''_{K_{0} M_{3} } (x)}}} \right. \kern-\nulldelimiterspace} {{f}''_{K_{0} M_{3} } (x)}$, then we have
\begin{align}
& g_{K_{0} h\mathunderscore K_{0} M_{3} } (x)=\notag\\
& =\frac{3\left( {x-1} \right)^{2}\left( {x+1} \right)\sqrt {2x+2} }{\left(
{\begin{array}{l}
 \sqrt {2x+2} \left( {x+1} \right)\times \\
 \times \left( {3x^{2}+2x+3} \right)-16x\left( {x^{{3/2}}+1} \right) \\
 \end{array}} \right)}\notag
\end{align}
and
\begin{align}
& {g}'_{K_{0} h\mathunderscore K_{0} M_{3} } (x)=\notag\\
& =-\frac{24\left( {x-1} \right)\sqrt {2x+2} \times k_{27} (x)}{\left(
{\begin{array}{l}
 \left[ {\sqrt {2x+2} \left( {x+1} \right)\times } \right. \\
 \left. {\times \left( {3x^{2}+2x+3} \right)-16x\left( {x^{(3/2)}+1}
\right)} \right]^{2} \\
 \end{array}} \right)},\notag
\end{align}
where $k_{27} (x)=k_{25} (x)>0$ This gives
\[
{g}'_{K_{0} h\mathunderscore K_{0} M_{3} } (x)\begin{cases}
 {>0,} & {x<1} \\
 {<0,} & {x>1} \\
\end{cases},
\]
Also we have

\[
\beta =\mathop {\sup }\limits_{x\in (0,\infty
)} g_{K_{0} h\mathunderscore K_{0} M_{3} } (x)=
\lim\limits_{x\to 1} g_{K_{0} h\mathunderscore K_{0} M_{3} }
(x)=\textstyle{4 \over 3}.
\]

\item \textbf{For }$\bf{D_{K_{0} h\,}^{32} (P\vert \vert Q)\le \textstyle{3 \over 2}D_{K_{0} J}^{30} (P\vert \vert Q)}$\textbf{: }Let us consider $g_{K_{0} h\mathunderscore K_{0} J} (x)={{f}''_{K_{0} h} (x)} \mathord{\left/ {\vphantom {{{f}''_{K_{0} h} (x)} {{f}''_{K_{0} J} (x)}}} \right. \kern-\nulldelimiterspace} {{f}''_{K_{0} J} (x)}$, then we have
\[
g_{K_{0} h\mathunderscore K_{0} J} (x)=\frac{3\left( {\sqrt x +1}
\right)^{2}}{3x+2\sqrt x +3},
\]
\begin{align}
& {g}'_{K_{0} h\mathunderscore K_{0} J} (x)=\notag\\
& =-\frac{6\left( {x-1} \right)}{\sqrt x \left( {3x+2\sqrt x +3}
\right)^{2}}\begin{cases}
 {>0,} & {x<1} \\
 {<0,} & {x>1} \\
\end{cases}\notag
\end{align}
and
\[
\beta =\mathop {\sup }\limits_{x\in (0,\infty )}
g_{K_{0} h\mathunderscore K_{0} J} (x)=
\lim\limits_{x\to 1} g_{K_{0} h\mathunderscore K_{0} J} (x)=\textstyle{3
\over 2}.
\]

\item \textbf{For }$\bf{D_{K_{0} h\,}^{32} (P\vert \vert Q)\le \textstyle{1 \over 4}D_{\Psi \Delta }^{45} (P\vert \vert Q)}$\textbf{: }Let us consider $g_{K_{0} h\mathunderscore \Psi \Delta } (x)={{f}''_{K_{0} h} (x)} \mathord{\left/ {\vphantom {{{f}''_{K_{0} h} (x)} {{f}''_{\Psi \Delta } (x)}}} \right. \kern-\nulldelimiterspace} {{f}''_{\Psi \Delta } (x)}$, then we have
\[
g_{K_{0} h\mathunderscore \Psi \Delta } (x)=\frac{3\sqrt x \left( {x+1}
\right)^{3}}{4\left( {x^{4}+5x^{3}+12x^{2}+5x+1} \right)},
\]
\begin{align}
& {g}'_{K_{0} h\mathunderscore \Psi \Delta } (x)= \notag\\
& =-\frac{\left( {\begin{array}{l}
 8(x-1)^{3}(x+1)^{2}\times \\
 \times (x^{2}+5x+1) \\
 \end{array}} \right)}{3\sqrt x \left( {\begin{array}{l}
 x^{4}+5x^{3}+ \\
 +12x^{2}+5x+1 \\
 \end{array}} \right)^{2}}\begin{cases}
 {>0,} & {x<1} \\
 {<0,} & {x>1} \\
\end{cases}\notag
\end{align}
and
\[
\beta=\mathop {\sup }\limits_{x\in (0,\infty
)} g_{K_{0} h\mathunderscore \Psi \Delta } (x)=
\lim\limits_{x\to 1} g_{K_{0} h\mathunderscore \Psi \Delta }
(x)=\textstyle{1 \over 4}.
\]

\item \textbf{For }$\bf{D_{K_{0} J\,}^{30} (P\vert \vert Q)\le \textstyle{1 \over 5}D_{\Psi I}^{44} (P\vert \vert Q)}$\textbf{: }Let us consider $g_{K_{0} J\mathunderscore \Psi \Delta } (x)={{f}''_{K_{0} J} (x)} \mathord{\left/ {\vphantom {{{f}''_{K_{0} J} (x)} {{f}''_{\Psi \Delta } (x)}}} \right. \kern-\nulldelimiterspace} {{f}''_{\Psi \Delta } (x)}$, then we have
\[
g_{K_{0} J\mathunderscore \Psi I} (x)=\frac{\sqrt x \left( {3x+2\sqrt x +3}
\right)\left( {x+1} \right)}{4\left( {x^{2}+3x+1} \right)\left( {\sqrt x +1}
\right)^{2}}
\]
and
\begin{align}
& {g}'_{K_{0} J\mathunderscore \Psi I} (x)=-\frac{\left( {\sqrt x -1}
\right)\left( {3x+\sqrt x +3} \right)}{8\sqrt x \left( {x^{2}+3x+1}
\right)^{2}\left( {\sqrt x +1} \right)^{3}}\times \notag\\
& \times \left( {\begin{array}{l}
 \left( {x^{5/2}+1} \right)\left( {\sqrt x +1} \right)+ \\
 +2x\left( {x+1} \right)+x\left( {\sqrt x -1} \right)^{2} \\
 \end{array}} \right).\notag
\end{align}

This gives
\[
{g}'_{K_{0} J\mathunderscore \Psi I} (x)\begin{cases}
 {>0,} & {x<1} \\
 {<0,} & {x>1} \\
\end{cases}.
\]
Also we have
\[
\beta =\mathop {\sup }\limits_{x\in (0,\infty )}
g_{K_{0} J\mathunderscore \Psi I} (x)=
\lim\limits_{x\to 1} g_{K_{0} J\mathunderscore \Psi I} (x)=\textstyle{1
\over 5}.
\]

\item \textbf{For }$\bf{D_{\Psi \Delta }^{45} (P\vert \vert Q)\le \textstyle{6 \over 5}D_{\Psi I}^{44} (P\vert \vert Q)}$\textbf{: }It is true in view of (\ref{eq2}).

\bigskip
\item \textbf{For }$\bf{D_{K_{0} M_{3} }^{31} (P\vert \vert Q)\le \textstyle{3 \over {13}}D_{\Psi M_{1} }^{43} (P\vert \vert Q)}$\textbf{: }Let us consider $g_{K_{0} M_{3} \mathunderscore \Psi M_{1} } (x)={{f}''_{K_{0} M_{3} } (x)} \mathord{\left/ {\vphantom {{{f}''_{K_{0} M_{3} } (x)} {{f}''_{\Psi M_{1} } (x)}}} \right. \kern-\nulldelimiterspace} {{f}''_{\Psi M_{1} } (x)}$, then we have
\begin{align}
& g_{K_{0} M_{3} \mathunderscore \Psi M_{1} } (x)=\notag\\
& =\frac{\sqrt x \left( {\begin{array}{l}
 \sqrt {2x+2} \left( {x+1} \right)\times \\
 \times \left( {3x^{2}+2x+3} \right)-16x\left( {x^{3/2}+1} \right) \\
 \end{array}} \right)}{4\left( {\begin{array}{l}
 \sqrt {2x+2} \left( {x+1} \right)\times \\
 \times \left( {x^{3}-4x^{3/2}+1} \right)+4x^{3/2}\left( {x^{3/2}+1} \right)
\\  \end{array}} \right)}\notag
\end{align}
and
\begin{align}
& {g}'_{K_{0} M_{3} \mathunderscore \Psi M_{1} } (x)=\notag\\
& =-\frac{3\sqrt x \sqrt {2x+2} \left( {\sqrt x -1} \right)\times k_{28}
(x)}{\left( {\begin{array}{l}
 8\sqrt x \left( {\sqrt x -1} \right)^{3}\left( {x+1} \right)^{2}\sqrt
{2x+2} \times \\
 \times \left( {4x^{2}+5x^{3/2}+6x+5\sqrt x +4} \right)^{2} \\
 \end{array}} \right)},\notag
\end{align}
where
\begin{align}
& k_{28} (x)=\sqrt {2x+2} \left( {\sqrt x +1} \right)\times \notag\\
& \times \left( {x+1} \right)^{2}\left( {\begin{array}{l}
 x^{4}+3x^{3}+8x^{5/2}+ \\
 +8x^{2}+8x^{3/2}+3x+1 \\
 \end{array}} \right)-\notag\\
& -4x\left( {\begin{array}{l}
 4x^{5}+6x^{9/2}+6x^{4}+13x^{7/2}+ \\
 +18x^{3}+34x^{5/2}+18x^{2}+ \\
 +13x^{3/2}+6x+6\sqrt x +4 \\
 \end{array}} \right).\notag
\end{align}

This gives
\[
{g}'_{K_{0} M_{3} \mathunderscore \Psi M_{1} } (x)\begin{cases}
 {>0,} & {x<1} \\
 {<0,} & {x>1} \\
\end{cases},
\]
provided $k_{28} (x)>0$. In order to prove $k_{28} (x)>0$, let us consider
\begin{align}
& h_{28} (x)=\left[ {\sqrt {2x+2} \left( {\sqrt x +1} \right)}
\right]^{2}\times \notag\\
& \times \left[ {\left( {x+1} \right)^{2}\left( {\begin{array}{l}
 x^{4}+3x^{3}+8x^{5/2}+ \\
 +8x^{2}+8x^{3/2}+3x+1 \\
 \end{array}} \right)} \right]^{2} -\notag\\
& -\left[ {4x\left( {\begin{array}{l}
 4x^{5}+6x^{9/2}+6x^{4}+13x^{7/2}+ \\
 +18x^{3}+34x^{5/2}+18x^{2}+ \\
 +13x^{3/2}+6x+6\sqrt x +4 \\
 \end{array}} \right)} \right]^{2}.\notag
\end{align}

After simplification, we have
\begin{align}
& h_{28} (x)=2\left( {\sqrt x -1} \right)^{4}\times \notag\\
& \times \left( {\begin{array}{l}
 1+8\sqrt x +45x+1760x^{5/2}+698x^{2}+ \\
 +3616x^{3}+10243x^{4}+18541x^{6}+ \\
 +18541x^{5}+3616x^{8}+10243x^{7}+ \\
 +x^{11}+698x^{9}+45x^{10}+8x^{21/2}+ \\
 +6416x^{7/2}+14712x^{9/2}+208x^{3/2}+ \\
 +14712x^{13/2}+1760x^{17/2}+ \\
 +6416x^{15/2}+208x^{19/2}+20112x^{11/2} \\
 \end{array}} \right).\notag
\end{align}

Since $h_{28} (x)>0$, proving that $k_{28} (x)>0$. Also we have
\[
\beta=\mathop {\sup }\limits_{x\in
(0,\infty )} g_{K_{0} M_{3} \mathunderscore \Psi M_{1} } (x)=
\]
\[
\quad = \lim\limits_{x\to 1} g_{K_{0} M_{3} \mathunderscore \Psi M_{1} }
(x)=\textstyle{3 \over {13}}.
\]

\item \textbf{For }$\bf{D_{\Psi I}^{44} (P\vert \vert Q)\le \textstyle{{40} \over {39}}D_{\Psi M_{1} }^{43} (P\vert \vert Q)}$\textbf{: }Let us consider $g_{\Psi I\mathunderscore \Psi M_{1} } (x)={{f}''_{\Psi I} (x)} \mathord{\left/ {\vphantom {{{f}''_{\Psi I} (x)} {{f}''_{\Psi M_{1} } (x)}}} \right. \kern-\nulldelimiterspace} {{f}''_{\Psi M_{1} } (x)}$, then we have
\begin{align}
& g_{\Psi I\mathunderscore \Psi M_{1} } (x)= \notag\\
& =\frac{\left( {x^{2}+3x+1} \right)\left( {x-1} \right)^{2}\sqrt {2x+2}
}{\left( {\begin{array}{l}
 \sqrt {2x+2} \left( {x^{3}-4x^{3/2}+1} \right)\times \\
 \times \left( {x+1} \right)+4x^{3/2}\left( {x^{3/2}+1} \right) \\
 \end{array}} \right)}\notag
\end{align}
and
\begin{align}
& {g}'_{\Psi I\mathunderscore \Psi M_{1} } (x)=\notag\\
& =-\frac{2x^{3/2}\sqrt {2x+2} (x-1)\times k_{29} (x)}{\left( {x+1}
\right)\left( {\begin{array}{l}
 \sqrt {2x+2} \left( {x^{3}-4x^{3/2}+1} \right)\times \\
 \times \left( {x+1} \right)+4x^{3/2}\left( {x^{3/2}+1} \right) \\
 \end{array}} \right)^{2}},\notag
\end{align}
where
\begin{align}
& k_{29} (x)=\sqrt {2x+2} \left( {x+1} \right)\times \notag\\
& \times \left( {\begin{array}{l}
 3\left( {\sqrt x -1} \right)^{2}\left( {x^{2}+1} \right)\left( {x+1}
\right)+ \\
 +2x^{7/2}+3x^{3}+10x^{2}+3x+2\sqrt x \\
 \end{array}} \right)- \notag\\
& -\left( {\begin{array}{l}
 3x^{11/2}+6x^{9/2}+6x^{4}+10x^{7/2}+15x^{3}+ \\
 +15x^{5/2}+10x^{2}+6x^{3/2}+6x+3 \\
 \end{array}} \right).\notag
\end{align}

This gives
\[
{g}'_{\Psi I\mathunderscore \Psi M_{1} } (x)\begin{cases}
 {>0} & {x<1} \\
 {<0} & {x>1} \\
\end{cases},
\]
provides $k_{29} (x)>0$. In order to prove $k_{29} (x)>0$, let us consider
\begin{align}
& h_{29} (x)=\left[ {\sqrt {2x+2} \left( {x+1} \right)} \right]^{2}\times \notag\\
& \times \left( {\begin{array}{l}
 3\left( {\sqrt x -1} \right)^{2}\left( {x^{2}+1} \right)\left( {x+1}
\right)+ \\
 +2x^{7/2}+3x^{3}+10x^{2}+3x+2\sqrt x \\
 \end{array}} \right)^{2}- \notag\\
& -\left( {\begin{array}{l}
 3x^{11/2}+6x^{9/2}+6x^{4}+10x^{7/2}+15x^{3}+ \\
 +15x^{5/2}+10x^{2}+6x^{3/2}+6x+3 \\
 \end{array}} \right)^{2}. \notag
\end{align}

After simplifications, we have
\[
h_{29} (x)=\left( {\sqrt x -1} \right)^{4}\times s(x),
\]
where
\begin{align}
& s(x)=\left( {x+1} \right)\left( {\begin{array}{l}
 9+56x+179x^{2}+360x^{3}+ \\
 +491x^{4}+491x^{5}+360x^{6}+ \\
 +179x^{7}+56x^{8}+9x^{9} \\
 \end{array}} \right)- \notag\\
& -2\sqrt x \left( {\begin{array}{l}
 6+32x+81x^{2}+144x^{3}+179x^{4}+ \\
 +144x^{5}+81x^{6}+32x^{7}+6x^{8} \\
 \end{array}} \right). \notag
\end{align}

We know that $x+1\ge 2 \sqrt x$, this allows us to conclude that
\begin{align}
& s(x)\ge  2 \sqrt x \left[ {\left( {\begin{array}{l}
 9+56x+179x^{2}+360x^{3}+ \\
 +491x^{4}+491x^{5}+360x^{6}+ \\
 +179x^{7}+56x^{8}+9x^{9} \\
 \end{array}} \right)} \right.- \notag\\
& \left. {-\left( {\begin{array}{l}
 6+32x+81x^{2}+144x^{3}+179x^{4}+ \\
 +144x^{5}+81x^{6}+32x^{7}+6x^{8} \\
 \end{array}} \right)} \right]. \notag
\end{align}

After simplifications, we have
\[
s(x)\ge 2 \sqrt x \left( {\begin{array}{l}
 3x^{8}+15x^{7}+51x^{6}+ \\
 +84x^{5}+84x^{4}+84x^{3}+ \\
 +51x^{2}+15x+3 \\
 \end{array}} \right)\ge 0.
\]
This gives us $h_{29} (x)>0$. Hence $k_{29} (x)>0$.
Also we have
\[
\beta = \mathop {\sup }\limits_{x\in (0,\infty
)} g_{\Psi I\mathunderscore \Psi M_{1} } (x)
=\lim\limits_{x\to 1} g_{\Psi I\mathunderscore \Psi M_{1} }
(x)=\textstyle{{40} \over {39}}.
\]

\item \textbf{For }$\bf{D_{\Psi M_{1} }^{43} (P\vert \vert Q)\le \textstyle{{39} \over {37}}D_{\Psi M_{2} }^{42} (P\vert \vert Q)}$\textbf{: }Let us consider $g_{\Psi M_{1} \mathunderscore \Psi M_{2} } (x)={{f}''_{\Psi M_{1} } (x)} \mathord{\left/ {\vphantom {{{f}''_{\Psi M_{1} } (x)} {{f}''_{\Psi M_{2} } (x)}}} \right. \kern-\nulldelimiterspace} {{f}''_{\Psi M_{2} } (x)}$, then we have
\begin{align}
& g_{\Psi M_{1} \mathunderscore \Psi M_{2} } (x)= \notag\\
& =\frac{3\left( {\begin{array}{l}
 \sqrt {2x+2} \left( {x^{3}-4x^{3/2}+1} \right)\times \\
 \times \left( {x+1} \right)+4x^{3/2}\left( {x^{3/2}+1} \right) \\
 \end{array}} \right)}{\left( {\begin{array}{l}
 \sqrt {2x+2} \left( {3x^{3}-8x^{3/2}+3} \right)\times \\
 \times \left( {x+1} \right)+4x^{2}\left( {x^{3/2}+1} \right) \\
 \end{array}} \right)}\notag
\end{align}
and
\begin{align}
& {g}'_{\Psi M_{1} \mathunderscore \Psi M_{2} } (x)=\notag\\
& =-\frac{36\sqrt x \left( {x^{3/2}-1} \right)\left( {x+1} \right)\times
k_{30} (x)}{\left( {\begin{array}{l}
 \sqrt {2x+2} \left[ {\sqrt {2x+2} \left( {x+1} \right)\times } \right. \\
 \left. {\times \left( {3x^{3}-8x^{3/2}+3} \right)+4x^{3/2}\left(
{x^{3/2}+1} \right)} \right]^{2} \\
 \end{array}} \right)},\notag
\end{align}
where
\begin{align}
& k_{30} (x)=2\left( {x^{4}+3x^{5/2}+3x^{3/2}+1} \right)^{2}- \notag\\
& -\sqrt {2x+2} \left( {x^{3/2}+1} \right)\left( {x+1} \right)^{2}. \notag
\end{align}

This gives
\[
{g}'_{\Psi M_{1} \mathunderscore \Psi M_{2} } (x)\begin{cases}
 {>0,} & {x<1} \\
 {<0,} & {x>1} \\
\end{cases},
\]
provided $k_{30} (x)>0$. In order to prove $k_{30} (x)>0$, let us consider
\begin{align}
& h_{30} (x)=\left[ {2\left( {x^{4}+3x^{5/2}+3x^{3/2}+1} \right)} \right]^{2}-\notag\\
& -\left[ {\sqrt {2x+2} \left( {x^{3/2}+1} \right)\left( {x+1} \right)^{2}}
\right]^{2}. \notag
\end{align}

After simplifications, we have
\begin{align}
& h_{30} (x)=2\left( {\sqrt x -1} \right)^{4}\times \notag\\
& \times \left( {\begin{array}{l}
 x^{6}+4x^{11/2}+5x^{5}+10x^{9/2}+15x^{4}+ \\
 +18x^{7/2}+24x^{3}+18x^{5/2}+15x^{2}+ \\
 +10x^{3/2}+5x+4\sqrt x +1 \\
 \end{array}} \right).\notag
\end{align}

Since $h_{30} (x)>0$, this gives that $k_{30} (x)>0$. Also we have
\[
\beta =\mathop {\sup }\limits_{x\in
(0,\infty )} g_{\Psi M_{1} \mathunderscore \Psi M_{2} } (x)=
\lim\limits_{x\to 1} g_{\Psi M_{1} \mathunderscore \Psi M_{2} }
(x)=\textstyle{{39} \over {37}}.
\]

\item \textbf{For }$D_{\Psi M_{2} }^{42} (P\vert \vert Q)\le \textstyle{{37} \over {36}}D_{\Psi h}^{41} (P\vert \vert Q)$\textbf{: }Let us consider $g_{\Psi M_{2} \mathunderscore \Psi h} (x)={{f}''_{\Psi M_{2} } (x)} \mathord{\left/ {\vphantom {{{f}''_{\Psi M_{2} } (x)} {{f}''_{\Psi h} (x)}}} \right. \kern-\nulldelimiterspace} {{f}''_{\Psi h} (x)}$, then we have
\begin{align}
& g_{\Psi M_{2} \mathunderscore \Psi h} (x)= \notag\\
& =\frac{\left( {\begin{array}{l}
 \sqrt {2x+2} \left( {3x^{3}-8x^{3/2}+3} \right)\times \\
 \times \left( {x+1} \right)+4x^{3/2}\left( {x^{3/2}+1} \right) \\
 \end{array}} \right)}{3\left[ \sqrt x {\left( {x+1} \right)\sqrt {2x+2} \left(
{x^{3/2}-1} \right)^{2}} \right]} \notag
\end{align}
and
\[
{g}'_{\Psi M_{2} \mathunderscore \Psi h} (x)=-\frac{\sqrt x \times k_{31}
(x)}{\sqrt {2x+2} \left( {x+1} \right)^{2}\left( {x^{3/2}-1} \right)^{3}},
\]
where $k_{31} (x)=k_{30} (x)>0$.This gives
\[
{g}'_{\Psi M_{2} \mathunderscore \Psi h} (x)\begin{cases}
 {>0,} & {x<1} \\
 {<0,} & {x>1} \\
\end{cases}.
\]
Also we have
\[
\beta =\mathop {\sup }\limits_{x\in (0,\infty
)} g_{\Psi M_{2} \mathunderscore \Psi h} (x)=
\lim\limits_{x\to 1} g_{\Psi M_{2} \mathunderscore \Psi h}
(x)=\textstyle{{37} \over {36}}.
\]

\item \textbf{For }$\bf{D_{\Psi h}^{41} (P\vert \vert Q)\le \textstyle{9 \over 8}D_{\Psi J}^{39} (P\vert \vert Q)}$\textbf{: }It is true in view of (\ref{eq2}).

\bigskip
\item \textbf{For}$\bf{D_{\Psi h}^{41} (P\vert \vert Q)\le \textstyle{{12} \over {11}}D_{\Psi M_{3} }^{40} (P\vert \vert Q)}$\textbf{: }Let us consider $g_{\Psi h\mathunderscore \Psi M_{3} } (x)={{f}''_{\Psi h} (x)} \mathord{\left/ {\vphantom {{{f}''_{\Psi h} (x)} {{f}''_{\Psi M_{3} } (x)}}} \right. \kern-\nulldelimiterspace} {{f}''_{\Psi M_{3} } (x)}$, then we have
\begin{align}
& g_{\Psi h\mathunderscore \Psi M_{3} } (x)= \notag\\
& =\frac{\left( {x^{3/2}-1} \right)^{2}\left( {x+1} \right)\sqrt {2x+2}
}{\left( {\begin{array}{l}
 \sqrt {2x+2} \left( {x^{3}+1} \right)\times \\
 \times \left( {x+1} \right)-4x^{3/2}\left( {x^{3/2}+1} \right) \\
 \end{array}} \right)} \notag
\end{align}
and
\begin{align}
& {g}'_{\Psi h\mathunderscore \Psi M_{3} } (x)= \notag\\
& =-\frac{3x^{3/2}\sqrt {2x+2} \left( {x^{3/2}-1} \right)\times k_{32}
(x)}{\left( {\begin{array}{l}
 \sqrt {2x+2} \left( {x^{3}+1} \right)\times \\
 \times \left( {x+1} \right)-4x^{3/2}\left( {x^{3/2}+1} \right) \\
 \end{array}} \right)^{2}},\notag
\end{align}

where $k_{32} (x)=k_{30} (x)>0$. This gives
\[
{g}'_{\Psi h\mathunderscore \Psi M_{3} } (x)\begin{cases}
 {>0,} & {x<1} \\
 {<0,} & {x>1} \\
\end{cases}.
\]
Also we have
\[
\beta =\mathop {\sup }\limits_{x\in (0,\infty
)} g_{\Psi h\mathunderscore \Psi M_{3} } (x)=
\lim\limits_{x\to 1} g_{\Psi h\mathunderscore \Psi M_{3} }
(x)=\textstyle{{12} \over {11}}.
\]

\item \textbf{For }$\bf{D_{\Psi M_{3} }^{40} (P\vert \vert Q)\le \textstyle{{11} \over 8}D_{\Psi K_{0} }^{37} (P\vert \vert Q)}$\textbf{: }Let us consider $g_{\Psi M_{3} \mathunderscore \Psi K_{0} } (x)={{f}''_{\Psi M_{3} } (x)} \mathord{\left/ {\vphantom {{{f}''_{\Psi M_{3} } (x)} {{f}''_{\Psi K_{0} } (x)}}} \right. \kern-\nulldelimiterspace} {{f}''_{\Psi K_{0} } (x)}$, then we have
\begin{align}
& g_{\Psi M_{3} \mathunderscore \Psi K_{0} } (x)= \notag\\
& =\frac{4\left( {\begin{array}{l}
 \sqrt {2x+2} \left( {x^{3}+1} \right)\times \\
 \times \left( {x+1} \right)-4x^{3/2}\left( {x^{3/2}+1} \right) \\
 \end{array}} \right)}{\left( {\begin{array}{l}
 \left( {x+1} \right)\sqrt {2x+2} \left( {\sqrt x -1} \right)^{2}\times \\
 \times \left( {4x^{2}+5x^{3/2}+6x+5\sqrt x +4} \right) \\
 \end{array}} \right)} \notag
\end{align}
and
\begin{align}
& {g}'_{\Psi M_{3} \mathunderscore \Psi K_{0} } (x)= \notag\\
& =-\frac{6\times k_{33} (x)}{\left( {\begin{array}{l}
 \left( {\sqrt x -1} \right)^{3}\sqrt {2x+2} \sqrt x \left( {x+1}
\right)^{2}\times \\
 \times \left( {4x^{2}+5x^{3/2}+6x+5\sqrt x +4} \right)^{2} \\
 \end{array}} \right)},\notag
\end{align}
where
\begin{align}
& k_{33} (x)=\left( {\begin{array}{l}
 \sqrt {2x+2} \left( {\sqrt x +1} \right)\left( {x+1} \right)^{2}\times \\
 \times \left( {1+3x+8x^{2}+3x^{3}+x^{4}} \right) \\
 \end{array}} \right)- \notag\\
& -4x\left( {\begin{array}{l}
 4+2x^{4}+2x+6x^{2}+11x^{3/2}+2\sqrt x + \\
 +14x^{5/2}+6x^{3}+11x^{7/2}+4x^{5}+2x^{9/2} \\
 \end{array}} \right).\notag
\end{align}

This gives
\[
{g}'_{\Psi M_{3} \mathunderscore \Psi K_{0} } (x)\begin{cases}
 {>0} & {x<1} \\
 {<0} & {x>1} \\
\end{cases},
\]
provided $k_{33} (x)>0$. In order to prove $k_{33} (x)>0$, let us consider
\begin{align}
& h_{33} (x)=\left[ {\begin{array}{l}
 \sqrt {2x+2} \left( {\sqrt x +1} \right)\left( {x+1} \right)^{2}\times \\
 \times \left( {1+3x+8x^{2}+3x^{3}+x^{4}} \right) \\
 \end{array}} \right]^{2}- \notag\\
& -\left[ {4x\left( {\begin{array}{l}
 4+2x^{4}+2x+6x^{2}+11x^{3/2}+ \\
 +2\sqrt x +14x^{5/2}+6x^{3}+ \\
 +11x^{7/2}+4x^{5}+2x^{9/2} \\
 \end{array}} \right)} \right]^{2}.\notag
\end{align}

After simplifications, we have
\begin{align}
& h_{33} (x)=\left( {\sqrt x -1} \right)^{4}\times \notag\\
& \times \left( {\begin{array}{l}
 1+30x+2144x^{5}+6\sqrt x +231x^{2}+ \\
 +698x^{3}+380x^{5/2}+110x^{3/2}+ \\
 +1056x^{7/2}+1475x^{4}+1874x^{9/2}+ \\
 +2298x^{6}+698x^{9}+30x^{11}+231x^{10}+ \\
 +2144x^{7}+380x^{19/2}+1056x^{17/2}+ \\
 +1475x^{8}+x^{12}+6x^{23/2}+110x^{21/2}+ \\
 +2366x^{13/2}+2366x^{11/2}+1874x^{15/2} \\
 \end{array}} \right). \notag
\end{align}

Since $h_{33} (x)>0$, this gives that $k_{33} (x)>0$. Also we have
\[
\beta = \mathop {\sup }\limits_{x\in (0,\infty )} g_{\Psi M_{3} \mathunderscore \Psi
K_{0} } (x)=\lim\limits_{x\to 1} g_{\Psi M_{3} \mathunderscore \Psi K_{0}
} (x)=\textstyle{{11} \over 8}.
\]

\item \textbf{For }$\bf{D_{\Psi J}^{39} (P\vert \vert Q)\le \textstyle{4 \over 3}D_{\Psi K_{0} }^{37} (P\vert \vert Q)}$\textbf{: }Let us consider $g_{\Psi J\mathunderscore \Psi K_{0} } (x)={{f}''_{\Psi J} (x)} \mathord{\left/ {\vphantom {{{f}''_{\Psi J} (x)} {{f}''_{\Psi K_{0} } (x)}}} \right. \kern-\nulldelimiterspace} {{f}''_{\Psi K_{0} } (x)}$, then we have
\[
g_{\Psi J\mathunderscore \Psi K_{0} } (x)=\frac{4\left( {\sqrt x +1}
\right)^{2}\left( {x+1} \right)}{4x^{2}+5x^{3/2}+6x+5\sqrt x +4},
\]
\begin{align}
& {g}'_{\Psi J\mathunderscore \Psi K_{0} } (x)=\notag\\
& =-\frac{2\left( {x-1} \right)\left( {\begin{array}{l}
 3x^{2}+4x^{3/2}+ \\
 +10x+4\sqrt x +3 \\
 \end{array}} \right)}{\sqrt x \left( {\begin{array}{l}
 4x^{2}+5x^{3/2}+ \\
 +6x+5\sqrt x +4 \\
 \end{array}} \right)^{2}}\begin{cases}
 {>0,} & {x<1} \\
 {<0,} & {x>1} \\
\end{cases}\notag
\end{align}
and
\[
\beta =\mathop {\sup }\limits_{x\in (0,\infty )} g_{\Psi J\mathunderscore \Psi K_{0}
} (x)=\lim\limits_{x\to 1} g_{\Psi J\mathunderscore \Psi K_{0} }
(x)=\textstyle{4 \over 3}.
\]

\item \textbf{For }$\bf{D_{\Psi K_{0} }^{37} (P\vert \vert Q)\le D_{\Psi T}^{38} (P\vert \vert Q)}$\textbf{: }It is true in view of pyramid.

\bigskip
\item \textbf{For }$\bf{D_{\Psi K_{0} }^{37} (P\vert \vert Q)\le \textstyle{1 \over 3}D_{F\Delta }^{55} (P\vert \vert Q)}$\textbf{: }Let us consider $g_{\Psi K_{0} \mathunderscore F\Delta } (x)={{f}''_{\Psi K_{0} } (x)} \mathord{\left/ {\vphantom {{{f}''_{\Psi K_{0} } (x)} {{f}''_{F\Delta } (x)}}} \right. \kern-\nulldelimiterspace} {{f}''_{F\Delta } (x)}$, then we have
\begin{align}
& g_{\Psi K_{0} \mathunderscore F\Delta } (x)= \notag\\
& =\frac{4\sqrt x \left( {\begin{array}{l}
 4x^{2}+5x^{3/2}+ \\
 +6x+5\sqrt x +4 \\
 \end{array}} \right)\left( {x+1} \right)^{3}}{\left( {\begin{array}{l}
 15+90x+257x^{2}+492x^{3}+ \\
 +257x^{4}90x^{5}+15x^{6}+30\sqrt x + \\
 +150x^{3/2}+364x^{5/2}+ \\
 +364x^{7/2}+150x^{9/2}+30x^{11/2} \\
 \end{array}} \right)},\notag
 \end{align}
\begin{align}
& {g}'_{\Psi K_{0} \mathunderscore F\Delta } (x)= \notag\\
& =-\frac{12\left( {x+1} \right)^{2}\left( {x-1} \right)^{3}\left( {\sqrt x
-1} \right)^{2}}{\sqrt x \left( {\begin{array}{l}
 15+90x+257x^{2}+492x^{3}+ \\
 +257x^{4}+90x^{5}+15x^{6}+ \\
 +30\sqrt x +150x^{3/2}+364x^{5/2}+ \\
 +364x^{7/2}+150x^{9/2}+30x^{11/2} \\
 \end{array}} \right)^{2}}\times \notag\\
& \times \left( {\begin{array}{l}
 10+110x+486x^{2}+ \\
 +740x^{3}+486x^{4}+ \\
 +110x^{5}+10x^{6}+\,25\sqrt x + \\
 +205x^{3/2}+586x^{5/2}+ \\
 +586x^{7/2}+205x^{9/2}+25x^{11/2} \\
 \end{array}} \right)\begin{cases}
 {>0,} & {x<1} \\
 {<0,} & {x>1} \\
\end{cases} \notag
\end{align}
and
\[
\beta =\mathop {\sup }\limits_{x\in (0,\infty )} g_{\Psi K_{0} \mathunderscore
F\Delta } (x)=\lim\limits_{x\to 1} g_{\Psi K_{0} \mathunderscore F\Delta
} (x)=\textstyle{1 \over 3}.
\]

\item \textbf{For }$\bf{D_{\Psi T}^{38} (P\vert \vert Q)\le \textstyle{3 \over 8}D_{FI}^{54} (P\vert \vert Q)}$\textbf{: }Let us consider \newline $g_{\Psi T\mathunderscore FI} (x)={{f}''_{\Psi T} (x)} \mathord{\left/ {\vphantom {{{f}''_{\Psi T} (x)} {{f}''_{FI} (x)}}} \right. \kern-\nulldelimiterspace} {{f}''_{FI} (x)}$, then we have
\begin{align}
& g_{\Psi T\mathunderscore FI} (x)= \notag\\
& =\frac{16\sqrt x \left( {x^{2}+x+1} \right)\left( {\sqrt x +1}
\right)^{2}}{\left( {\begin{array}{l}
 15x^{4}+30x^{7/2}+60x^{3}+90x^{5/2}+ \\
 +122x^{2}+90x^{3/2}+60x+30\sqrt x +15 \\
 \end{array}} \right)} \notag
\end{align}
and
\begin{align}
& {g}'_{\Psi T\mathunderscore FI} (x)= \notag\\
& =-\frac{8\left( {x-1} \right)}{\sqrt x \left( {\begin{array}{l}
 15x^{4}+30x^{7/2}+60x^{3}+ \\
 +90x^{5/2}+122x^{2}+90x^{3/2}+ \\
 +60x+30\sqrt x +15 \\
 \end{array}} \right)^{2}}\times \notag\\
& \times \left( {\begin{array}{l}
 15x^{6}+26x^{11/2}+40x^{5}+86x^{9/2}+9x^{4}+ \\
 +9x^{2}+86x^{3/2}+40x+26\sqrt x +15+ \\
 +(x^{2}-1)^{2}\sqrt x \left( {34x+65\sqrt x +34} \right) \\
 \end{array}} \right). \notag
\end{align}

This gives
\[
g_{\Psi T\mathunderscore FI} (x)\begin{cases}
 {>0,} & {x<1} \\
 {<0,} & {x>1} \\
\end{cases}.
\]
Also, we have
\[
\beta =\mathop {\sup }\limits_{x\in (0,\infty )} g_{\Psi T\mathunderscore FI}
(x)=\lim\limits_{x\to 1} g_{\Psi T\mathunderscore FI} (x)=\textstyle{3 \over 8}.
\]

\item \textbf{For }$\bf{D_{F\Delta }^{55} (P\vert \vert Q)\le \textstyle{9 \over 8}D_{FI}^{54} (P\vert \vert Q)}$\textbf{: }Let us consider \newline $g_{\Psi \Delta \mathunderscore FI} (x)={{f}''_{\Psi \Delta } (x)} \mathord{\left/ {\vphantom {{{f}''_{\Psi \Delta } (x)} {{f}''_{FI} (x)}}} \right. \kern-\nulldelimiterspace} {{f}''_{FI} (x)}$, then we have
\begin{align}
& g_{F\Delta \mathunderscore FI} (x)= \notag\\
& =\frac{\left( {\begin{array}{l}
 15+90x+257x^{2}+492x^{3}+ \\
 +257x^{4}+90x^{5}+15x^{6}+ \\
 +30\sqrt x +150x^{3/2}+364x^{5/2}+ \\
 +364x^{7/2}+150x^{9/2}+30x^{11/2} \\
 \end{array}} \right)}{\left( {x+1} \right)^{2}\left( {\begin{array}{l}
 15x^{4}+30x^{7/2}+60x^{3}+ \\
 +90x^{5/2}+122x^{2}+90x^{3/2}+ \\
 +60x+30\sqrt x +15 \\
 \end{array}} \right)} \notag
\end{align}
\begin{align}
& {g}'_{F\Delta \mathunderscore FI} (x)= \notag\\
& =-\frac{32x^{3/2}\left( {x-1} \right)}{\left( {x+1} \right)^{3}\left(
{\begin{array}{l}
 15x^{4}+30x^{7/2}+60x^{3}+ \\
 +90x^{5/2}+122x^{2}+90x^{3/2}+ \\
 +60x+30\sqrt x +15 \\
 \end{array}} \right)^{2}}\times \notag\\
& \times \left( {\begin{array}{l}
 75x^{5}+300x^{9/2}+675x^{4}+1200x^{7/2}+ \\
 +1682x^{3}+1928x^{5/2}+1682x^{2}+ \\
 +1200x^{3/2}+675x+300\sqrt x +75 \\
 \end{array}} \right). \notag
\end{align}

This gives
\[
{g}'_{F\Delta \mathunderscore FI} (x)\begin{cases}
 {>0,} & {x<1} \\
 {<0,} & {x>1} \\
\end{cases}.
\]
Also, we have
\[
\beta =\mathop {\sup }\limits_{x\in (0,\infty )} g_{F\Delta \mathunderscore FI}
(x)=\lim\limits_{x\to 1} g_{F\Delta \mathunderscore FI} (x)=\textstyle{9
\over 8}.
\]

\item \textbf{For }$\bf{D_{FI}^{54} (P\vert \vert Q)\le \textstyle{{64} \over {63}}D_{FM_{1} }^{53} (P\vert \vert Q)}$\textbf{: }Let us consider $g_{FI\mathunderscore FM_{1} } (x)={{f}''_{FI} (x)} \mathord{\left/ {\vphantom {{{f}''_{FI} (x)} {{f}''_{FM_{1} } (x)}}} \right. \kern-\nulldelimiterspace} {{f}''_{FM_{1} } (x)}$, then we have
\begin{align}
& g_{FI\mathunderscore FM_{1} } (x)= \notag\\
& =\frac{\left( {\sqrt x -1} \right)^{2}\sqrt {2x+2} }{\left(
{\begin{array}{l}
 \sqrt {2x+2} \left( {15x^{4}-62x^{2}+15} \right)\times \\
 \times \left( {x+1} \right)+64x^{2}\left( {x^{3/2}+1} \right) \\
 \end{array}} \right)}\times \notag\\
& \times \left( {\begin{array}{l}
 15x^{4}+30x^{7/2}+60x^{3}+90x^{5/2}+ \\
 +122x^{2}+90x^{3/2}+60x+30\sqrt x +15 \\
 \end{array}} \right) \notag
\end{align}
and
\begin{align}
& {g}'_{FI\mathunderscore FM_{1} } (x)= \notag\\
& =-\frac{64\left( {\sqrt x -1} \right)\times k_{34} (x)}{\left(
{\begin{array}{l}
 \sqrt {2x+2} \left( {15x^{4}-62x^{2}+15} \right)\times \\
 \times \left( {x+1} \right)+64x^{2}\left( {x^{3/2}+1} \right) \\
 \end{array}} \right)^{2}}, \notag
\end{align}
where
\begin{align}
& k_{34} (x)=\sqrt {2x+2} \,\times \notag\\
& \times \left( {\begin{array}{l}
 -15x^{6}+242x^{5/2}+165x+ \\
 +302x^{3}+45x^{5}+242x^{4}+ \\
 +225x^{2}+225x^{9/2}+60+ \\
 +45x^{3/2}+60x^{13/2}-15\sqrt x + \\
 +302x^{7/2}+165x^{11/2} \\
 \end{array}} \right)- \notag\\
& -\left( {\begin{array}{l}
 60+135x+255x^{5}+60\sqrt x +255x^{2}+ \\
 +478x^{3}+484x^{5/2}+240x^{3/2}+ \\
 +672x^{7/2}+478x^{4}+484x^{9/2}+ \\
 +135x^{6}+60x^{13/2}+240x^{11/2}+60x^{7} \\
 \end{array}} \right). \notag
\end{align}

This gives
\[
{g}'_{FI\mathunderscore FM_{1} } (x)\begin{cases}
 {>0,} & {x<1} \\
 {<0,} & {x>1} \\
\end{cases},
\]
provided $k_{34} (x)>0$. In order to prove $k_{34} (x)>0$, let us consider
\begin{align}
& h_{34} (x)=\left( {\sqrt {2x+2} } \right)^{2}\times \notag\\
& \times \left( {\begin{array}{l}
 -15x^{6}+242x^{5/2}+165x+302x^{3}+ \\
 +45x^{5}+242x^{4}+225x^{2}+225x^{9/2}+ \\
 +60+45x^{3/2}+60x^{13/2}-15\sqrt x + \\
 +302x^{7/2}+165x^{11/2} \\
 \end{array}} \right)^{2}- \notag\\
& -\left( {\begin{array}{l}
 60+135x+255x^{5}+60\sqrt x +255x^{2}+ \\
 +478x^{3}+484x^{5/2}+240x^{3/2}+ \\
 +672x^{7/2}+478x^{4}+484x^{9/2}+135x^{6}+ \\
 +60x^{13/2}+240x^{11/2}+60x^{7} \\
 \end{array}} \right)^{2}. \notag
\end{align}

After simplification, we get
\begin{align}
& h_{34} (x)=\left( {\sqrt x -1} \right)^{4}\times \notag\\
& \times \left( {\begin{array}{l}
 3600+441250x^{6}+20250x+ \\
 +396832x^{5}+3600\sqrt x +61875x^{2}+ \\
 +147270x^{3}+90000x^{5/2}+26100x^{3/2}+ \\
 +208260x^{7/2}+277740x^{4}+ \\
 +352860x^{9/2}+20250x^{11}+ \\
 +451980x^{11/2}+451980x^{13/2}+ \\
 +352860x^{15/2}+3600x^{23/2}+3600x^{12}+ \\
 +147270x^{9}+396832x^{7}+ \\
 +26100x^{21/2}+277740x^{8}+ \\
 +208260x^{17/2}+90000x^{19/2}+61875x^{10} \\
 \end{array}} \right). \notag
\end{align}

Since $h_{34} (x)>0$, this gives that $k_{34} (x)>0$. Also we have
\[
\beta =\mathop {\sup }\limits_{x\in (0,\infty )} g_{FI\mathunderscore FM_{1} }
(x)=\lim\limits_{x\to 1} g_{FI\mathunderscore FM_{1} }
(x)=\textstyle{{64} \over {63}}.
\]

\item \textbf{For }$\bf{D_{FM_{1} }^{53} (P\vert \vert Q)\le \textstyle{{63} \over {61}}D_{FM_{2} }^{52} (P\vert \vert Q)}$\textbf{: }Let us consider $g_{FM_{1} \mathunderscore FM_{2} } (x)={{f}''_{FM_{1} } (x)} \mathord{\left/ {\vphantom {{{f}''_{FM_{1} } (x)} {{f}''_{FM_{2} } (x)}}} \right. \kern-\nulldelimiterspace} {{f}''_{FM_{2} } (x)}$, then we have
\begin{align}
& g_{FM_{1} \mathunderscore FM_{2} } (x)= \notag\\
& =\frac{3\left( {\begin{array}{l}
 \sqrt {2x+2} \left( {15x^{4}-62x^{2}+15} \right)\times \\
 \times \left( {x+1} \right)+64x^{2}\left( {x^{3/2}+1} \right) \\
 \end{array}} \right)}{\left( {\begin{array}{l}
 \sqrt {2x+2} \left( {45x^{4}-122x^{2}+45} \right)\times \\
 \times \left( {x+1} \right)+64x^{2}\left( {x^{3/2}+1} \right) \\
 \end{array}} \right)} \notag
\end{align}
and
\begin{align}
& {g}'_{FM_{1} \mathunderscore FM_{2} } (x)=-\frac{5760x\left( {x-1}
\right)\left( {x+1} \right)^{3}}{\sqrt {2x+2} }\times \notag\\
& \times \frac{k_{35} (x)}{\left( {\begin{array}{l}
 \sqrt {2x+2} \left( {45x^{4}-122x^{2}+45} \right)\times \\
 \times \left( {x+1} \right)+64x^{2}\left( {x^{3/2}+1} \right) \\
 \end{array}} \right)^{2}}, \notag
\end{align}
where
\begin{align}
& k_{35} (x)=\left( {\sqrt x +1} \right)\times \notag\\
& \times \left[ {x^{2}\left( {2\sqrt x -1} \right)^{2}+6x^{3/2}+\left( {\sqrt
x -2} \right)^{2}} \right]- \notag\\
& -2\sqrt {2x+2} \left( {x^{2}+1} \right)\left( {x+1} \right). \notag
\end{align}

This gives
\[
{g}'_{FM_{1} \mathunderscore FM_{2} } (x)\begin{cases}
 {>0,} & {x<1} \\
 {<0,} & {x>1} \\
\end{cases},
\]
provided $k_{35} (x)>0$. In order to prove $k_{35} (x)>0$, let us consider
\begin{align}
& h_{35} (x)=\left( {\sqrt x +1} \right)^{2}\times \notag\\
& \times \left[ {x^{2}\left( {2\sqrt x -1} \right)^{2}+6x^{3/2}+\left( {\sqrt
x -2} \right)^{2}} \right]^{2}- \notag\\
& -\left[ {2\sqrt {2x+2} \left( {x^{2}+1} \right)\left( {x+1} \right)}
\right]^{2}. \notag
\end{align}

After simplifications, we have
\begin{align}
& h_{35} (x)=\left( {\sqrt x -1} \right)^{4}\times \notag\\
& \times \left( {\begin{array}{l}
 8x^{5}+32x^{9/2}+32x^{4}+24x^{7/2}+ \\
 +49x^{3}+82x^{5/2}+49x^{2}+ \\
 +24x^{3/2}+32x+32\sqrt x +8 \\
 \end{array}} \right). \notag
\end{align}

Since $h_{35} (x)>0$, this gives that $k_{35} (x)>0$. Also we have
\[
\beta =\mathop {\sup }\limits_{x\in (0,\infty )} g_{FM_{1} \mathunderscore FM_{2} }
(x)=\lim\limits_{x\to 1} g_{FM_{1} \mathunderscore FM_{2} }
(x)=\textstyle{{63} \over {61}}.
\]

\item \textbf{For }$\bf{D_{FM_{2} }^{52} (P\vert \vert Q)\le \textstyle{{61} \over {60}}D_{Fh}^{51} (P\vert \vert Q)}$\textbf{: }Let us consider $g_{FM_{2} \mathunderscore Fh} (x)={{f}''_{FM_{2} } (x)} \mathord{\left/ {\vphantom {{{f}''_{FM_{2} } (x)} {{f}''_{Fh} (x)}}} \right. \kern-\nulldelimiterspace} {{f}''_{Fh} (x)}$, then we consider
\begin{align}
& g_{FM_{2} \mathunderscore Fh} (x)= \notag\\
& =\frac{\left( {\begin{array}{l}
 \sqrt {2x+2} \left( {45x^{4}-122x^{2}+45} \right)\times \\
 \times \left( {x+1} \right)+64x^{2}\left( {x^{3/2}+1} \right) \\
 \end{array}} \right)}{45\sqrt {2x+2} \left( {x+1} \right)^{3}\left( {x-1}
\right)^{2}} \notag
\end{align}
and
\[
{g}'_{FM_{2} \mathunderscore Fh} (x)=-\frac{32x\times k_{36} (x)}{45\left(
{x-1} \right)^{3}\sqrt {2x+2} \left( {x+1} \right)^{4}}
\]
where $k_{36} (x)=k_{35} (x)>0$. This gives
\[
{g}'_{FM_{2} \mathunderscore Fh} (x)\begin{cases}
 {>0} & {x<1} \\
 {<0} & {x>1} \\
\end{cases}.
\]
Also we have
\[
\beta =\mathop {\sup }\limits_{x\in (0,\infty )} g_{FM_{2} \mathunderscore Fh}
(x)=\lim\limits_{x\to 1} g_{FM_{2} \mathunderscore Fh}
(x)=\textstyle{{61} \over {60}}.
\]

\item \textbf{For} $\bf{D_{Fh}^{51} (P\vert \vert Q)\le \textstyle{{15} \over {14}}D_{FJ}^{49} (P\vert \vert Q)}$\textbf{: }Let us consider \newline $g_{Fh\mathunderscore FJ} (x)={{f}''_{Fh} (x)} \mathord{\left/ {\vphantom {{{f}''_{Fh} (x)} {{f}''_{FJ} (x)}}} \right. \kern-\nulldelimiterspace} {{f}''_{FJ} (x)}$, then we have
\[
g_{Fh\mathunderscore FJ} (x)=\frac{15\left( {\sqrt x +1} \right)^{2}\left(
{x+1} \right)^{2}}{\left( {\begin{array}{l}
 15x^{3}+30x^{5/2}+45x^{2}+ \\
 +44x^{3/2}+45x+30\sqrt x +15 \\
 \end{array}} \right)},
\]
and
\begin{align}
& {g}'_{Fh\mathunderscore FJ} (x)= \notag\\
& =-\frac{120\left( {x-1} \right)\left( {x+1} \right)\sqrt x \left( {3x+4\sqrt
x +3} \right)}{\left( {\begin{array}{l}
 15x^{3}+30x^{5/2}+45x^{2}+ \\
 +44x^{3/2}+45x+30\sqrt x +15 \\
 \end{array}} \right)^{2}}. \notag
\end{align}

This gives
\[
{g}'_{Fh\mathunderscore FJ} (x)\begin{cases}
 {>0,} & {x<1} \\
 {<0,} & {x>1} \\
\end{cases}.
\]
Also, we have
\[
\beta =\mathop {\sup }\limits_{x\in (0,\infty )} g_{Fh\mathunderscore FJ} (x)=\lim\limits_{x\to 1} g_{Fh\mathunderscore FJ} (x)=\textstyle{{15} \over {14}}.
\]

\item \textbf{For }$\bf{D_{Fh}^{51} (P\vert \vert Q)\le \textstyle{{20} \over {19}}D_{FM_{3} }^{50} (P\vert \vert Q)}$\textbf{: }Let us consider $g_{Fh\mathunderscore FM_{3} } (x)={{f}''_{Fh} (x)} \mathord{\left/ {\vphantom {{{f}''_{Fh} (x)} {{f}''_{FM_{3} } (x)}}} \right. \kern-\nulldelimiterspace} {{f}''_{FM_{3} } (x)}$, then we have
\begin{align}
& g_{Fh\mathunderscore FM_{3} } (x)= \notag\\
& =\frac{15\left( {x-1} \right)^{2}\left( {x+1} \right)^{3}\sqrt {2x+2}
}{\left( {\begin{array}{l}
 \sqrt {2x+2} \left( {15x^{4}+2x^{2}+15} \right)\times \\
 \times \left( {x+1} \right)-64x^{2}\left( {x^{3/2}+1} \right) \\
 \end{array}} \right)} \notag
\end{align}
and
\begin{align}
& {g}'_{Fh\mathunderscore FM_{3} } (x)= \notag\\
& =-\frac{960x\left( {x-1} \right)\left( {x+1} \right)^{2}\times k_{37}
(x)}{\left( {\begin{array}{l}
 \sqrt {2x+2} \left[ {-64x^{2}\left( {x^{3/2}+1} \right)} \right.+\sqrt
{2x+2} \times \\
 \left. {\times \left( {x+1} \right)\left( {15x^{4}+2x^{2}+15} \right)}
\right]^{2} \\
 \end{array}} \right)}, \notag
\end{align}
Where $k_{37} (x)=k_{35} (x)>0$. This gives
\[
{g}'_{Fh\mathunderscore FM_{3} } (x)\begin{cases}
 {>0,} & {x<1} \\
 {<0,} & {x>1} \\
\end{cases}.
\]
Also we have
\[
\beta =\mathop {\sup }\limits_{x\in (0,\infty )} g_{Fh\mathunderscore FM_{3} }
(x)=\lim\limits_{x\to 1} g_{Fh\mathunderscore FM_{3} }
(x)=\textstyle{{20} \over {19}}.
\]

\item \textbf{For }$\bf{D_{FM_{3} }^{50} (P\vert \vert Q)\le \textstyle{{19} \over {16}}D_{FK_{0} }^{47} (P\vert \vert Q)}$\textbf{: }Let us consider $g_{FM_{3} \mathunderscore FK_{0} } (x)={{f}''_{FM_{3} } (x)} \mathord{\left/ {\vphantom {{{f}''_{FM_{3} } (x)} {{f}''_{FK_{0} } (x)}}} \right. \kern-\nulldelimiterspace} {{f}''_{FK_{0} } (x)}$, then we have
\begin{align}
& g_{FM_{3} \mathunderscore FK_{0} } (x)= \notag\\
& =\frac{\left( {\begin{array}{l}
 \sqrt {2x+2} \left( {15x^{4}+2x^{2}+15} \right)\times \\
 \times \left( {x+1} \right)-64x^{2}\left( {x^{3/2}+1} \right) \\
 \end{array}} \right)}{3(x+1)(x-1)^{2}\sqrt {2x+2} \left( {5x^{2}+6x+5}
\right)} \notag
\end{align}
and
\[
{g}'_{FM_{3} \mathunderscore FK_{0} } (x)=-\frac{4\times k_{38} (x)}{\left(
{\begin{array}{l}
 3\left( {x-1} \right)^{3}\sqrt x \sqrt {2x+2} \times \\
 \times \left( {x+1} \right)^{2}\left( {5x^{2}+6x+5} \right)^{2} \\
 \end{array}} \right)},
\]
where
\begin{align}
& k_{38} (x)=\sqrt {2x+2} \left( {x+1} \right)^{3}\times \notag\\
& \times \left( {15x^{4}+20x^{3}+58x^{2}+20x+15} \right)- \notag\\
& -8x\left( {\sqrt x +1} \right)\left( {\begin{array}{l}
 10x^{5}+\left( {\sqrt x -1} \right)^{2}\times \\
 \left( {10x^{4}+x^{3}+x+10} \right)+ \\
 +26x^{4}+22x^{3}+12x^{5/2}+ \\
 +22x^{2}+26x+10 \\
 \end{array}} \right). \notag
\end{align}
This gives
\[
{g}'_{FM_{3} \mathunderscore FK_{0} } (x)\begin{cases}
 {>0,} & {x<1} \\
 {<0,} & {x>1} \\
\end{cases},
\]
provided $k_{38} (x)>0$. In order to prove $k_{38} (x)>0$, let us consider
\begin{align}
& h_{38} (x)=\left[ {\sqrt {2x+2} \left( {x+1} \right)^{3}} \right]^{2}\times \notag\\
& \times \left( {15x^{4}+20x^{3}+58x^{2}+20x+15} \right)^{2}- \notag\\
& -\left[ {8x\left( {\sqrt x +1} \right)\left( {\begin{array}{l}
 10x^{5}+\left( {\sqrt x -1} \right)^{2}\times \\
 \times \left( {10x^{4}+x^{3}+x+10} \right)+ \\
 +26x^{4}+22x^{3}+12x^{5/2}+ \\
 +22x^{2}+26x+10 \\
 \end{array}} \right)} \right]^{2}.\notag
\end{align}

After simplifications, we have
\begin{align}
& h_{38} (x)=2\left( {\sqrt x -1} \right)^{4}\times \notag\\
& \times \left( {\begin{array}{l}
 225+4425x+27890x^{2}+135760x^{7/2}+ \\
 +349215x^{5}+900\sqrt x +94290x^{3}+ \\
 +27890x^{11}+436976x^{13/2}+389920x^{15/2}+ \\
 +349215x^{8}+208791x^{9}+13200x^{3/2}+ \\
 +49160x^{5/2}+225x^{13}+265724x^{17/2}+ \\
 +208791x^{4}+453852x^{7}+389920x^{11/2}+ \\
 +4425x^{12}+94290x^{10}+3200x^{23/2}+ \\
 +900x^{25/2}+49160x^{21/2}+ \\
 +135760x^{19/2}+453852x^{6}+265724x^{9/2} \\
 \end{array}} \right). \notag
\end{align}

Since $h_{38} (x)>0$, proving that $k_{38} (x)>0$. Also we have
\[
\beta =\mathop {\sup }\limits_{x\in (0,\infty )} g_{FM_{3} \mathunderscore FK_{0} }
(x)=\lim\limits_{x\to 1} g_{FM_{3} \mathunderscore FK_{0} }
(x)=\textstyle{{19} \over {16}}.
\]

\item \textbf{For }$\bf{D_{FJ}^{49} (P\vert \vert Q)\le \textstyle{7 \over 6}D_{FK_{0} }^{47} (P\vert \vert Q)}$ \textbf{: } Let us consider $ g_{FJ\mathunderscore FK_{0} } (x)={{f}''_{FJ} (x)} \mathord{\left/ {\vphantom {{{f}''_{FJ} (x)} {{f}''_{FK_{0} } (x)}}} \right. \kern-\nulldelimiterspace} {{f}''_{FK_{0} } (x)}$, then we have
\begin{align}
& g_{FJ\mathunderscore FK_{0} } (x)= \notag\\
& =\frac{\left( {\begin{array}{l}
 15x^{3}+30x^{5/2}+45x^{2}+ \\
 +\,44x^{3/2}+45x+30\sqrt x +15 \\
 \end{array}} \right)}{3\left( {5x^{2}+6x+5} \right)\left( {\sqrt x +1}
\right)^{2}} \notag
\end{align}
\noindent and
\begin{align}
& {g}'_{FJ\mathunderscore FK_{0} } (x)= \notag\\
& =-\frac{4\left( {\sqrt x -1} \right)\left( {\begin{array}{l}
 15x^{3}+30x^{5/2}+65x^{2}+ \\
 +68x^{3/2}+\,65x+30\sqrt x +15 \\
 \end{array}} \right)}{3\left( {5x^{2}+6x+5} \right)^{2}\left( {\sqrt x +1}
\right)^{3}}. \notag
\end{align}
This gives
\[
{g}'_{FJ\mathunderscore FK_{0} } (x)\begin{cases}
 {>0,} & {x<1} \\
 {<0,} & {x>1} \\
\end{cases}.
\]
Also, we have
\[
\beta =\mathop {\sup }\limits_{x\in (0,\infty )} g_{FJ\mathunderscore FK_{0} }
(x)=\lim\limits_{x\to 1} g_{FJ\mathunderscore FK_{0} } (x)=\textstyle{7
\over 6}.
\]

\item \textbf{For }$\bf{D_{FK_{0} }^{47} (P\vert \vert Q)\le D_{FT}^{48} (P\vert \vert Q)}$\textbf{: }It holds in view of pyramid.

\bigskip
\item \textbf{For }$\bf{D_{FT}^{48} (P\vert \vert Q)\le 2D_{F\Psi }^{46} (P\vert \vert Q)}$\textbf{: }Let us consider \newline $g_{FK_{0} \mathunderscore FT} (x)={{f}''_{FK_{0} } (x)} \mathord{\left/ {\vphantom {{{f}''_{FK_{0} } (x)} {{f}''_{FT} (x)}}} \right. \kern-\nulldelimiterspace} {{f}''_{FT} (x)}$, then we have
\begin{align}
& g_{FT\mathunderscore F\Psi } (x)= \notag\\
& =\frac{\left( {\begin{array}{l}
 15x^{4}+30x^{7/2}+60x^{3}+58x^{5/2}+ \\
 +\,58x^{2}+58x^{3/2}+60x+30\sqrt x +15 \\
 \end{array}} \right)}{\left( {x+1} \right)\left( {\begin{array}{l}
 15x^{3}+14x^{5/2}+13x^{2}+ \\
 +\,12x^{3/2}+13x+14\sqrt x +15 \\
 \end{array}} \right)} \notag
\end{align}
and
\begin{align}
& {g}'_{FT\mathunderscore F\Psi } (x)= \notag\\
& =-\frac{8\left( {x-1} \right)}{\sqrt x \left( {x+1} \right)^{2}\left(
{\begin{array}{l}
 15x^{3}+14x^{5/2}+ \\
 +13x^{2}+12x^{3/2}+ \\
 +13x+14\sqrt x +15 \\
 \end{array}} \right)^{2}}\times \notag\\
& \times \left( {\begin{array}{l}
 15x^{6}+60x^{11/2}+105x^{5}+184x^{9/2}+ \\
 +265x^{4}+380x^{7/2}+382x^{3}+380x^{5/2}+ \\
 +265x^{2}+184x^{3/2}+105x+60\sqrt x +15 \\
 \end{array}} \right). \notag
\end{align}

This gives
\[
{g}'_{FT\mathunderscore F\Psi } (x)\begin{cases}
 {>0,} & {x<1} \\
 {<0,} & {x>1} \\
\end{cases}.
\]
Also,we have
\[
\beta =\mathop {\sup }\limits_{x\in (0,\infty )} g_{FT\mathunderscore F\Psi }
(x)=\lim\limits_{x\to 1} g_{FT\mathunderscore F\Psi } (x)=2.
\]

\item \textbf{For }$\bf{D_{Jh}^{17} (P\vert \vert Q)\le 4D_{JM_{3} }^{16} (P\vert \vert Q)}$\textbf{: }Let us consider \newline $g_{Jh\mathunderscore JM_{3} } (x)={{f}''_{Jh} (x)} \mathord{\left/ {\vphantom {{{f}''_{Jh} (x)} {{f}''_{JM_{3} } (x)}}} \right. \kern-\nulldelimiterspace} {{f}''_{JM_{3} } (x)}$, then we have
\[
g_{Jh\mathunderscore JM_{3} } (x)=\frac{\left( {\sqrt x -1}
\right)^{2}\left( {x+1} \right)\sqrt {2x+2} }{\left( {\begin{array}{l}
 \sqrt {2x+2} \left( {x+1} \right)^{2}- \\
 -4\sqrt x \left( {x^{3/2}+1} \right) \\
 \end{array}} \right)}
\]
and
\[
{g}'_{Jh\mathunderscore JM_{3} } (x)=\frac{2\left( {1-\sqrt x }
\right)\left( {x+1} \right)\times k_{39} (x)}{\left( {\begin{array}{l}
 \sqrt x \sqrt {2x+2} \left[ {\sqrt {2x+2} \times } \right. \\
 \left. {\times \left( {x+1} \right)^{2}-4\sqrt x \left( {x^{3/2}+1}
\right)} \right]^{2} \\
 \end{array}} \right)},
\]
where $k_{39} (x)=k_{21} (x)>0$. This gives
\[
{g}'_{Jh\mathunderscore JM_{3} } (x)\begin{cases}
 {>0,} & {x<1} \\
 {<0,} & {x>1} \\
\end{cases}.
\]
Also, we have
\[
\beta =\mathop {\sup }\limits_{x\in (0,\infty )} g_{Jh\mathunderscore JM_{3} }
(x)=\lim\limits_{x\to 1} g_{Jh\mathunderscore JM_{3} } (x)=4.
\]

\item \textbf{For }$\bf{D_{M_{1} I}^{2} (P\vert \vert Q)\le D_{JM_{3} }^{16} (P\vert \vert Q)}$\textbf{: }Let us consider $g_{M_{1} I\mathunderscore JM_{3} } (x)={{f}''_{M_{1} I} (x)} \mathord{\left/ {\vphantom {{{f}''_{M_{1} I} (x)} {{f}''_{JM_{3} } (x)}}} \right. \kern-\nulldelimiterspace} {{f}''_{JM_{3} } (x)}$, then we have
\begin{align}
& g_{M_{1} I\mathunderscore JM_{3} } (x)= \notag\\
& =\frac{4\sqrt x \left[ {\sqrt {2x+2} \left( {x-\sqrt x +1} \right)-\left(
{x^{3/2}+1} \right)} \right]}{\sqrt {2x+2} \left( {x+1} \right)^{2}-4\sqrt x
\left( {x^{3/2}+1} \right)} \notag
\end{align}
and
\begin{align}
& {g}'_{M_{1} I\mathunderscore JM_{3} } (x)=-\frac{2\sqrt x \left( {\sqrt x
-1} \right)^{3}}{\left( {x+1} \right)}\times \notag\\
& \times \frac{k_{40} (x)}{\left[ {\sqrt {2x+2} \left( {x+1}
\right)^{2}-4\sqrt x \left( {x^{3/2}+1} \right)} \right]^{2}}, \notag
\end{align}
where
\begin{align}
& k_{40} (x)=\sqrt {2x+2} \left( {\sqrt x +1} \right)\left( {x+1}
\right)^{2}- \notag\\
& -\left( {x^{3}+3x^{5/2}+8x^{3/2}+3\sqrt x +1} \right). \notag
\end{align}

This gives
\[
{g}'_{M_{1} I\mathunderscore JM_{3} } (x)\begin{cases}
 {>0,} & {x<1} \\
 {<0,} & {x>1} \\
\end{cases},
\]
provided $k_{40} (x)>0$. In order to prove $k_{40} (x)>0$, let us consider
\begin{align}
& h_{40} (x)=\left[ {\sqrt {2x+2} \left( {\sqrt x +1} \right)\left( {x+1}
\right)^{2}} \right]^{2}- \notag\\
& -\left( {x^{3}+3x^{5/2}+8x^{3/2}+3\sqrt x +1} \right)^{2}. \notag
\end{align}

After simplifications, we have
\begin{align}
& h_{40} (x)=\left( {\sqrt x -1} \right)^{4}\times \notag\\
& \times \left( {\begin{array}{l}
 x^{5}+4x^{9/2}+8x^{4}+20x^{7/2}+ \\
 +24x^{3}+34x^{5/2}+24x^{2}+ \\
 +20x^{3/2}+8x+8\sqrt x +1 \\
 \end{array}} \right). \notag
\end{align}

Since $h_{40} (x)>0$, proving that $k_{40} (x)>0$. Also we have
\[
\beta =\mathop {\sup }\limits_{x\in (0,\infty )} g_{M_{1} I\mathunderscore JM_{3} }
(x)=\lim\limits_{x\to 1} g_{M_{1} I\mathunderscore JM_{3} } (x)=1.
\]

\item \textbf{For }$\bf{D_{TM_{3} }^{23} (P\vert \vert Q)\le 9D_{JM_{3} }^{16} (P\vert \vert Q)}$\textbf{: }Let us consider $g_{TM_{3} \mathunderscore JM_{3} } (x)={{f}''_{TM_{3} } (x)} \mathord{\left/ {\vphantom {{{f}''_{TM_{3} } (x)} {{f}''_{JM_{3} } (x)}}} \right. \kern-\nulldelimiterspace} {{f}''_{JM_{3} } (x)}$, then we have
\begin{align}
& g_{TM_{3} \mathunderscore JM_{3} } (x)= \notag\\
& =\frac{2\left[ {\left( {x^{2}+1} \right)\sqrt {2x+2} -2\sqrt x \left(
{x^{3/2}+1} \right)} \right]}{\sqrt {2x+2} \left( {x+1} \right)^{2}-4\sqrt x
\left( {x^{3/2}+1} \right)} \notag
\end{align}
and
\begin{align}
& {g}'_{TM_{3} \mathunderscore JM_{3} } (x)= \notag\\
& =-\frac{2\sqrt {2x+2} \left( {x-1} \right)\times k_{41} (x)}{\left(
{\begin{array}{l}
 \left( {x+1} \right)\sqrt x \left[ {\sqrt {2x+2} \times } \right. \\
 \left. {\times \left( {x+1} \right)^{2}-4\sqrt x \left( {x^{3/2}+1}
\right)} \right]^{2} \\
 \end{array}} \right)},\notag
\end{align}
where
\begin{align}
& k_{41} (x)=\left( {x^{7/2}+3x^{5/2}+4x^{2}+4x^{3/2}+3x+1} \right)- \notag\\
& -2\sqrt x \sqrt {2x+2} \left( {x+1} \right)^{2}.\notag
\end{align}

This gives
\[
{g}'_{TM_{3} \mathunderscore JM_{3} } (x)\begin{cases}
 {>0,} & {x<1} \\
 {<0,} & {x>1} \\
\end{cases},
\]
provided $k_{41} (x)>0$. In order to prove $k_{41} (x)>0$, let us consider
\begin{align}
& h_{41} (x)=\left( {\begin{array}{l}
 x^{7/2}+3x^{5/2}+4x^{2}+ \\
 +4x^{3/2}+3x+1 \\
 \end{array}} \right)^{2}- \notag\\
& -\left[ {2\sqrt x \sqrt {2x+2} \left( {x+1} \right)^{2}} \right]^{2}.\notag
\end{align}

After simplifications, we have
\begin{align}
& h_{41} (x)=\left( {\sqrt x -1} \right)^{4}\times \notag\\
& \times \left( {\begin{array}{l}
 x^{5}+4x^{9/2}+8x^{4}+20x^{7/2}+24x^{3}+ \\
 +34x^{5/2}+24x^{2}+20x^{3/2}+8x+4\sqrt x +1 \\
 \end{array}} \right). \notag
\end{align}

Since $h_{41} (x)>0$, proving that $k_{41} (x)>0$. Also we have
\[
\beta =\mathop {\sup }\limits_{x\in (0,\infty )} g_{TM_{3} \mathunderscore JM_{3} }
(x)=\lim\limits_{x\to 1} g_{TM_{3} \mathunderscore JM_{3} } (x)=9.
\]

\item \textbf{For }$\bf{D_{JM_{3} }^{16} (P\vert \vert Q)\le \textstyle{1 \over {24}}D_{F\Psi }^{46} (P\vert \vert Q)}$\textbf{: }Let us consider $g_{JM_{3} \mathunderscore F\Psi } (x)={{f}''_{JM_{3} } (x)} \mathord{\left/ {\vphantom {{{f}''_{JM_{3} } (x)} {{f}''_{F\Psi } (x)}}} \right. \kern-\nulldelimiterspace} {{f}''_{F\Psi } (x)}$, then we have
\begin{align}
& g_{JM_{3} \mathunderscore F\Psi } (x)= \notag\\
& =\frac{16x^{3/2}\left[ {\left( {x+1} \right)^{2}\sqrt {2x+2} -4\sqrt x
\left( {x^{3/2}+1} \right)} \right]}{\left( {\begin{array}{l}
 \sqrt {2x+2} \left( {x+1} \right)\left( {\sqrt x -1} \right)^{2}\times
\left( {15x^{3}+15} \right. \\
 \left. {+14x^{5/2}+13x^{2}+12x^{3/2}+13x+14\sqrt x } \right) \\
 \end{array}} \right)} \notag
\end{align}
and
\begin{align}
& {g}'_{JM_{3} \mathunderscore F\Psi } (x)=-\frac{8\sqrt x }{\sqrt {2x+2}
\left( {x+1} \right)^{2}\left( {\sqrt x -1} \right)^{3}}\times \notag\\
& \times \frac{k_{42} (x)}{\left( {\begin{array}{l}
 15x^{3}+15+14x^{5/2}+13x^{2}+ \\
 +12x^{3/2}+13x+14\sqrt x \\
 \end{array}} \right)}, \notag
\end{align}
where
\begin{align}
& k_{42} (x)=\sqrt {2x+2} \left( {x+1} \right)^{2}\times \notag\\
& \times \left( {\begin{array}{l}
 45x^{9/2}+13x^{4}+88x^{7/2}+24x^{3}+22x^{5/2}+ \\
 +22x^{2}+24x^{3/2}+88x+13\sqrt x +45 \\
 \end{array}} \right)- \notag\\
& -12\sqrt x \left( {\begin{array}{l}
 20x^{6}+4x^{11/2}+9x^{5}+44x^{9/2}+ \\
 +12x^{4}+32x^{7/2}+14x^{3}+32x^{5/2}+ \\
 +12x^{2}+44x^{3/2}+9x+4\sqrt x +20 \\
 \end{array}} \right). \notag
\end{align}

This gives
\[
{g}'_{JM_{3} \mathunderscore F\Psi } (x)\begin{cases}
 {>0,} & {x<1} \\
 {<0,} & {x>1} \\
\end{cases},
\]
provided $k_{42} (x)>0$. In order to prove $k_{42} (x)>0$, let us consider
\begin{align}
& h_{42} (x)=\left[ {\sqrt {2x+2} \left( {x+1} \right)^{2}} \right]^{2}\times \notag\\
& \times \left( {\begin{array}{l}
 45x^{9/2}+13x^{4}+88x^{7/2}+ \\
 +24x^{3}+22x^{5/2}+22x^{2}+ \\
 +24x^{3/2}+88x+13\sqrt x +45 \\
 \end{array}} \right)^{2}- \notag\\
& - 144\, x \left( {\begin{array}{l}
 20x^{6}+4x^{11/2}+9x^{5}+44x^{9/2}+ \\
 +12x^{4}+32x^{7/2}+14x^{3}+32x^{5/2}+ \\
 +12x^{2}+44x^{3/2}+9x+4\sqrt x +20 \\
 \end{array}} \right)^{2}. \notag
\end{align}

After simplifications, we have
\[
h_{42} (x)=2\left( {\sqrt x -1} \right)^{4}\times v(x),
\]
where
\begin{align}
& v(x)= \notag\\
& =\left( {\begin{array}{l}
 2025x^{12}+9270x^{23/2}+14344x^{11}+ \\
 +8634x^{21/2}+27498x^{10}+ \\
 +15106x^{19/2}+9952x^{9}-2034x^{17/2}- \\
 -9001x^{8}-9380x^{15/2}-12776x^{7}+ \\
 +1444x^{13/2}-36436x^{6}+1444x^{11/2}- \\
 -12776x^{5}-9380x^{9/2}-9001x^{4}- \\
 -2034x^{7/2}+9952x^{3}+15106x^{5/2}+ \\
 +27498x^{2}+8634x^{3/2}+ \\
 +14344x+9270\sqrt x +2025 \\
 \end{array}} \right). \notag
\end{align}

Now we shall show that $v(x)>0$. Let us consider
\begin{align}
& m(t)=v(t^{2})= \notag\\
& =\left( {\begin{array}{l}
 2025t^{24}+9270t^{23}+14344t^{22}+ \\
 +8634t^{21}+27498t^{20}+15106t^{19}+ \\
 +9952t^{18}-2034t^{17}-9001t^{16}- \\
 -9380t^{15}-12776t^{14}+1444t^{13}- \\
 -36436t^{12}+1444t^{11}-12776t^{10}- \\
 -9380t^{9}-9001t^{8}-2034t^{7}+ \\
 +9952t^{6}+15106t^{5}+27498t^{4}+ \\
 +8634t^{3}+14344t^{2}+9270t+2025 \\
 \end{array}} \right). \notag
\end{align}

The polynomial equation $m(t)=0$ of 24$^{\mathrm{th}}$ degree admits 24
solutions. Out of them 22 are complex (not written here) and two of them are
real given by
\begin{center}
$-1.125443752$ and $-0.8885384079$.
\end{center}
\noindent Both these solutions are negative. Since we are working with $t>0$, this means that there are no real positive solutions of the equation $m(t)=0$. Thus we conclude that either $m(t)>0$ or $m(t)<0$, for all $t>0$. In order to check it is sufficient to see for any particular value of $m(t)$, for example when $t=1$. This gives $m(1)=73728$, hereby proving that $m(t)>0$ for all $t>0$, consequently, $v(x)>0$, for all $x>0$, proving that $h_{42} (x)>0$, $\forall x>0,\,x\ne 1$. Since $h_{42} (x)>0$, proving that $k_{42} (x)>0$. Also we have
\[
\beta =\mathop {\sup }\limits_{x\in (0,\infty )} g_{JM_{3} \mathunderscore F\Psi }
(x)=\lim\limits_{x\to 1} g_{JM_{3} \mathunderscore F\Psi }
(x)= \textstyle{1 \over {24}}.
\]
\end{enumerate}

Parts 1-55 refers to the proof of the inequalities given in (\ref{eq7}) and the parts 56-59 give the proof of (\ref{eq8}). Combining the parts 1-59 we get the proof of the Theorem 2.1.
\end{proof}

\subsection{Remark}
\begin{enumerate}
\item Theorem 2.1 connects 54 members out of 55 appearing in the pyramid. Since some them are equals by multiplicative constants, the Theorem 2.1 contains 47 different measures. In this way we can make a sequential inequality connecting 34 divergence measures.

\bigskip
\item From the inequalities given in (\ref{eq7}) and (\ref{eq8}), it is interesting to observe
that all the measures remain between $D_{I\Delta }^1 $ and $D_{F\Psi }^{46}
$, i.e., in between the first members of first and last line of the pyramid.

\bigskip
\item The last members of each line (corners members) of the pyramid are connected in an
increasing order, i.e.,
\begin{align}
D_{I\Delta }^1 & \le \textstyle{8 \over 9}D_{M_1 \Delta }^3 \le \textstyle{8
\over {11}}D_{M_2 \Delta }^6 \le \textstyle{2 \over 3}D_{h\Delta }^{10} \le\notag\\
& \le \textstyle{8 \over {15}}D_{M_3 \Delta }^{15} \le \textstyle{1 \over 2}D_{J\Delta }^{21} \le \textstyle{1 \over
3}D_{T\Delta }^{28} \le\notag\\
\label{eq13}
& \le \textstyle{1 \over 3}D_{K_0 \Delta }^{36} \le
\textstyle{1 \over 6}D_{\Psi \Delta }^{45} \le \textstyle{1 \over
9}D_{F\Delta }^{55}.
\end{align}
\end{enumerate}

\section{Equivalent Inequalities}

As a consequence of Theorem 2.1, the sequences of inequalities appearing in (\ref{eq7}) and (\ref{eq8}) can be written in an individual form. This means that the 59 results proving the Theorem 2.1 can be written in an equivalent form. This we have done below in two groups. The first group is with four measures in each case and the second group is with three measures.

\subsection*{Group 1}
\begin{enumerate}
\item $\hspace{15pt} 80 M_{1}+16 M_{3} \le \Delta +20h ;$
\item $\hspace{15pt} \Delta +32 h \le 4T+128 M_{1};$
\item $\hspace{15pt} 6\Delta +256M_{2} \le 192I+3K_{0} ;$
\item $\hspace{15pt} 288M_{1} +224M_{2} \le 168I+9J;$
\item $\hspace{15pt} 12M_{1} +20M_{2} \le 5I+3T;$
\item $\hspace{15pt} 9J+256M_{2} \le 192I+72T;$
\item $\hspace{15pt} 10T+32M_{2} \le 3J+10h;$
\item $\hspace{15pt} 72I+128T\le 9K_{0} +512M_{3} ;$
\item $\hspace{15pt} 8I+4J\le K_{0} +32h;$
\item $\hspace{15pt} 4\Delta +8K_{0} \le \Psi +64h;$
\item $\hspace{15pt} 16I+10K_{0} \le \Psi +10J;$
\item $\hspace{15pt} 26K_{0} +192 M_{1} \le 3\Psi + 832 M_{3};$
\item $\hspace{15pt} 32M_{1} +32M_{3} \le J+8T;$
\item $\hspace{15pt} 4\Delta +3\Psi \le F+6K_{0} ;$
\item $\hspace{15pt} 48I+8\Psi \le 3F+128T;$
\item $\hspace{15pt} 48J+\Psi \le 2F+1536M_{3}.$
\end{enumerate}

\subsection*{Group 2}

\renewcommand{\arraystretch}{1.5}
\noindent\begin{tabular}{llll}
1. & $ I\le \frac{\Delta +128M_{1} }{36};$         &  21.  & $ h\le \frac{F+1280M_{2} }{976};$        \\
2. & $ I\le \frac{4\Delta +K_{0} }{24};$           &  22.  & $ h\ge \frac{3I+16M_{3} }{7};$           \\
3. & $ I\le \frac{20\Delta +\Psi }{96};$           &  23.  & $ M_{3} \le \frac{2\Delta +15J}{512};$   \\
4. & $ I\le \frac{32\Delta +F}{144};$              &  24.  & $ M_{3} \le \frac{3J+8h}{128};$          \\
5. & $ M_{1} \le \frac{\Delta +24M_{2} }{88};$     &  25.  & $ M_{3} \le \frac{T+3h}{16};$            \\
6. & $ M_{1} \le \frac{120I+K_{0} }{512};$         &  26.  & $ M_{3} \le \frac{K_{0} +24h}{128};$     \\
7. & $ M_{1} \le \frac{624I+\Psi }{2560};$         &  27.  & $ M_{3} \le \frac{F+304h}{1280};$        \\
8. & $ M_{1} \le \frac{1008I+F}{4096};$            &  28.  & $ M_{3} \le \frac{\Psi +176h}{768}$      \\
9. & $ M_{1} \ge \frac{3I+2M_{2} }{18};$           &  29.  & $ J\le \frac{K_{0} +16h}{3};$            \\
10.& $ M_{2} \le \frac{\Delta +44h}{64};$          &  30.  & $ J\le \frac{2 \Delta +16T}{3};$         \\
11.& $ M_{2} \le \frac{T+26M_{1} }{10};$           &  31.  & $ J\le \frac{\Psi +128h}{18};$           \\
12.& $ M_{2} \le \frac{K_{0} +208M_{1} }{80};$     &  32.  & $ J\le \frac{F+224h}{30};$               \\
13.& $ M_{2} \le \frac{\Psi +1184M_{1} }{416};$    &  33.  & $ J\ge \frac{120T+256M_{2} }{39};$       \\
14.& $ M_{2} \le \frac{F+1952M_{1} }{672};$        &  34.  & $ J\ge \frac{8T+256M_{3} }{9}.$          \\
15.& $ M_{2} \ge \frac{3I+9h}{16};$                &  35.  & $ K_{0} \le \frac{6J+\Psi }{8};$         \\
16.& $ h\le \frac{\Delta +64M_{3} }{20};$          &  36.  & $ K_{0} \le \frac{12J+F}{14};$           \\
17.& $ h\le \frac{3J+128M_{2} }{120};$             &  39.  & $ K_{0} \le \frac{3\Psi +512M_{3} }{22};$\\
18.& $ h\le \frac{T+16M_{2} }{13};$                &  38.  & $ K_{0} \le \frac{3F+1024M_{3} }{38};$   \\
19.& $ h\le \frac{K_{0} +128M_{2} }{104};$         &  39.  & $ \Psi \le \frac{F+16T}{2}.$             \\
20.& $ h\le \frac{\Psi +768M_{2} }{592};$          &
\end{tabular}

\bigskip
Direct relations of the inequalities given in Groups 1 and 2 to the inequalities given in (\ref{eq5}) shall be dealt elsewhere.

\section{Reverse Inequalities}

In view of Theorem 2.1, we shall derive some inequalities in reverse order for the last three lines of the \textbf{pyramid}.

\begin{enumerate}
\item Combining the inequalities given in the 10$^{\mathrm{th}}$ line of the pyramid and the one given in (\ref{eq7}) having the measure $F(P\vert \vert Q)$, we have the following extended inequality
\begin{align}
& D_{F\Psi }^{46} \le D_{FK_{0} }^{47} \le D_{FT}^{48} \le D_{FJ}^{49} \le
D_{FM_{3} }^{50} \le D_{Fh}^{51} \le\notag\\
& \le D_{FM_{2} }^{52} \le D_{FM_{1} }^{53} \le D_{FI}^{54} \le D_{F\Delta
}^{55} \le \textstyle{9 \over 8}D_{FI}^{54} \le\notag\\
& \le \textstyle{8 \over 7}D_{FM_{1} }^{53} \le \textstyle{{72} \over
{61}}D_{FM_{2} }^{52} \le \textstyle{6 \over 5}D_{Fh}^{51} \le \left\{
{\begin{array}{l}
 \textstyle{9 \over 7}D_{FJ}^{49} \\
 \textstyle{{24} \over {19}}D_{FM_{3} }^{50} \\
 \end{array}} \right\}\le \notag\\
\label{eq14}
& \le \textstyle{3 \over 2}D_{FK_{0} }^{47} \le \textstyle{3 \over
2}D_{FT}^{48} \le 3D_{F\Psi }^{46}.
\end{align}

According to inequalities given in pyramid we have $D_{FJ}^{49} \le
D_{FM_{3} }^{50} $ but according to our approach we don't have \textit{reverse relation} among the
measures $D_{FJ}^{49} $ and $D_{FM_{3} }^{50} $. Also $D_{Fh}^{51} $ is
related to $D_{FJ}^{49} $ and $D_{FM_{3} }^{50} $with different
multiplicative constants. We call the expression (3.20) as \textit{reverse inequalities}

\bigskip
\item Combining the inequalities given in the 9$^{\mathrm{th}}$ line of the pyramid and the one given in (\ref{eq7}) having the measure $\Psi (P\vert \vert Q)$, we have the following extended inequality
\[
D_{\Psi K_{0} }^{37} \le D_{\Psi T}^{38} \le D_{\Psi J}^{39} \le D_{\Psi
M_{3} }^{40} \le D_{\Psi h}^{41} \le D_{\Psi M_{2} }^{42} \le
\]
\[
\le D_{\Psi M_{1} }^{43} \le D_{\Psi I}^{44} \le D_{\Psi \Delta }^{45} \le
\textstyle{6 \over 5}D_{\Psi I}^{44} \le \textstyle{{16} \over {15}}D_{\Psi
M_{1} }^{43} \le
\]
\[
\le \textstyle{{48} \over {37}}D_{\Psi M_{2} }^{42} \le \textstyle{4 \over
3}D_{\Psi h}^{41} \le \left\{ {\begin{array}{l}
 \textstyle{3 \over 2}D_{\Psi J}^{39} \\
 \textstyle{{16} \over {11}}D_{\Psi M_{3} }^{40} \\
 \end{array}} \right\}\le
\]
\begin{equation}
\label{eq15}
\le 2D_{\Psi K_{0} }^{37} \le 2D_{\Psi T}^{38}.
\end{equation}
According to inequalities given in pyramid we have $D_{\Psi J}^{39} \le
D_{\Psi M_{3} }^{40} $ but according to our approach we don't have reverse
relation among the measures $D_{\Psi J} $ and $D_{\Psi M_{3} } $. Also
$D_{\Psi h}^{41} $ is related to $D_{\Psi J}^{39} $ and $D_{\Psi M_{3}
}^{40} $ with different multiplicative constants. Again we call the
expression (\ref{eq15}) as \textit{reverse inequalities}

\bigskip
\item Combining the inequalities given in the 8$^{\mathrm{th}}$ line of the pyramid and the one given in (\ref{eq7}) having the measure $\Psi (P\vert \vert Q)$, we have the following extended inequality
\begin{align}
& D_{K_{0} T}^{29} \le D_{K_{0} J}^{30} \le D_{K_{0} M_{3} }^{31} \le D_{K_{0}
h}^{32} \le \notag\\
& \hspace{5pt} \le D_{K_{0} M_{2} }^{33} \le D_{K_{0} M_{1} }^{34} \le D_{K_{0} I}^{35} \le D_{K_{0} \Delta }^{36} \le \notag\\
& \hspace{10pt} \le \textstyle{3 \over 2}D_{K_{0} I}^{35} \le \textstyle{8 \over 5}D_{K_{0}
M_{1} }^{34} \le \textstyle{{24} \over {13}}D_{K_{0} M_{2} }^{33}\le \notag\\
\label{eq16}
& \hspace{15pt} \le 2D_{K_{0} h}^{32}
\le \left\{ {\begin{array}{l}
 3D_{K_{0} J}^{30} \\
 \textstyle{8 \over 3}D_{K_{0} M_{3} }^{31} \\
 \end{array}} \right\}.
\end{align}

We observe that the measure $D_{K_{0} T}^{29} $ don't appears in the reverse
side. Moreover, it don't appears in Theorem 2.1 too.
\end{enumerate}

Similarly we can write \textit{reverse inequalities} for the other lines of the pyramid.

\end{multicols}
\end{document}